\documentclass[conference,10pt]{IEEEtran}
\usepackage{setspace}

\ifCLASSINFOpdf
\else
\fi
\usepackage{amsthm}
\usepackage{amsmath}
\usepackage{amssymb}
\usepackage{txfonts}
\usepackage{xcolor,paralist,times}
\usepackage{graphicx}
\usepackage{epstopdf}
\usepackage{multirow}
\usepackage{colortbl}
\usepackage{epstopdf}
\usepackage{hyperref}
\hypersetup{
    colorlinks,
    citecolor=brown,
    filecolor=black,
    linkcolor=black,
    urlcolor=brown
}

\newcommand\set[1]{\{#1\}}
\newcommand\sig{\mathit{sig}}
\newcommand\acc{\mathit{acc}}
\newcommand\rej{\mathit{rej}}
\newtheorem{Thm}{Theorem}
\newtheorem{Lem}[Thm]{Lemma}
\newtheorem{Def}{Definition}

\newtheorem{Exm}{Example}

\newtheorem{Cor}[Thm]{Corollary}

\newtheorem{Rule}{Rule}
\usepackage{algorithm}
\usepackage{algorithmic}
\usepackage{url}
\usepackage{hyperref}
\usepackage{breakurl}

\begin{document}
%
\title{Making Streett Determinization Tight}

\author{\IEEEauthorblockN{Cong Tian}
\IEEEauthorblockA{
ICTT and ISN Laboratory \\ Xidian
University\\ Xi'an, 710071, P.R. China\\
ctian@mail.xidian.edu.cn}
\and
\IEEEauthorblockN{Wensheng Wang}
\IEEEauthorblockA{
ICTT and ISN Laboratory \\ Xidian
University\\ Xi'an, 710071, P.R. China\\
wswang@stu.xidian.edu.cn}

\and
\IEEEauthorblockN{Zhenhua Duan}
\IEEEauthorblockA{
ICTT and ISN Laboratory \\ Xidian
University\\ Xi'an, 710071, P.R. China\\
zhhduan@mail.xidian.edu.cn}

}


%


\maketitle

\begin{abstract}
Optimal determinization construction of Streett automata is an important research problem because it is indispensable in numerous applications such as decision problems for tree temporal logics, logic games and system synthesis.
This paper presents a transformation from nondeterministic Streett automata (NSA) with $n$ states and $k$ Streett pairs to equivalent deterministic Rabin transition automata (DRTA) with $n^{5n}(n!)^{n}$ states, $O(n^{n^2})$ Rabin pairs for $k=\omega(n)$ and $n^{5n}k^{nk}$ states, $O(k^{nk})$ Rabin pairs for $k=O(n)$. This improves the state of the art Streett determinization construction with $n^{5n}(n!)^{n+1}$ states, $O(n^2)$ Rabin pairs and $n^{5n}k^{nk}n!$ states, $O(nk)$ Rabin pairs, respectively. Moreover, deterministic parity transition automata (DPTA) are obtained with $3(n(n+1)-1)!(n!)^{n+1}$ states, $2n(n+1)$ priorities for $k=\omega(n)$ and $3(n(k+1)-1)!n!k^{nk}$ states, $2n(k+1)$ priorities for $k=O(n)$, which improves the best construction with $n^{n}(k+1)^{n(k+1)}(n(k+1)-1)!$ states, $2n(k+1)$ priorities.
Further, we prove a lower bound state complexity for determinization construction from NSA to deterministic Rabin (transition) automata  i.e. $n^{5n}(n!)^{n}$ for $k=\omega(n)$ and $n^{5n}k^{nk}$ for $k=O(n)$, which matches the state complexity of the proposed  determinization construction.
Besides, we put forward a lower bound state complexity for determinization construction from NSA to deterministic parity (transition) automata i.e. $2^{\Omega(n^2 \log n)}$ for $k=\omega(n)$ and $2^{\Omega(nk \log nk)}$ for $k=O(n)$, which is the same as the state complexity of the proposed determinization construction in the exponent.
\end{abstract}

\begin{IEEEkeywords}
Streett automata, Rabin automata, determinization, state complexity, lower bound.
\end{IEEEkeywords}

%
\IEEEpeerreviewmaketitle

\section{Introduction}

Streett automata \cite{Streett82}  are nearly the same as B\"{u}chi automata \cite{Buchi62} except for the acceptance condition.
They are exponentially more succinct than B\"{u}chi automata in encoding infinite behaviors of systems \cite{SV89}. As a result, Streett automata have an advantage in modeling  behaviors of concurrent and reactive systems \cite{CZ12}.

Determinization is one of the fundamental notions in automata theory. Given a nondeterministic automaton $\mathcal{A}$, determinization of $\mathcal{A}$ is the construction of another deterministic automaton $\mathcal{B}$ that recognizes the same language as $\mathcal{A}$ does. As for Streett automata,  determinization constructions have been investigated for decades.
In 1992, Safra introduced the first determinization construction for nondeterministic Streett automata (NSA) by using an innovative data structure known as \emph{Streett Safra trees} \cite{Safra92}. The states of the resulting deterministic automata are not sets of states, but tree structures. Safra's construction transforms a NSA with $n$ states and $k$ Streett pairs into a deterministic Rabin automaton (DRA) with $12^{n(k+1)}n^{n}(k+1)^{n(k+1)}(n(k+1))^{n(k+1)}$ states and $n(k+1)$ Rabin pairs. In 2007, Piterman \cite{Piterman07} presented a tighter construction via \emph{compact Streett Safra trees} which are obtained by using a dynamic naming technique throughout the Streett Safra tree construction. With compact Streett Safra trees, a NSA can be transformed into an equivalent deterministic parity automaton (DPA) with $2n^{n}(k+1)^{n(k+1)}(n(k+1))!$ states and $2n(k+1)$ priorities; or a DRA with the same state complexity and $n(k+1)$ Rabin pairs. The key advantage of Piterman's determinization  is the resulting DPA which is easier to manipulate. In 2012, Cai and Zhang presented the construction of an equivalent DRA with $n^{7n}(n!)^{n+1}$ states and $O(n^2)$ Rabin pairs for $k=\omega(n)$, and $n^{5n}k^{n(k+2)}n!$ states and $O(nk)$ Rabin pairs for $k=O(n)$ \cite{CZ12,CZ13}. Their construction is based on another data structure, namely, \emph{$\mu$-Safra trees for Streett determinization}, which reduces the redundancy of index labels and utilizes a batch-mode naming scheme.


As for the state lower bound of Streett determinization, it has also been investigated. For a NSA with $n$ states and $k$ Streett pairs, Cai and Zhang proved a lower bound of Streett complementation which is $2^{\Omega(n\log n+nk\log k)}$ states  for $k=O(n)$ and $2^{\Omega(n^2\log n)}$ states for $k=\omega(n)$ \cite{CZ11(2)}. It indicates that the lower bound state complexity for determinization construction from NSA to DR(T)A is no smaller than (maybe very close to) $2^{\Omega(n\log n+nk\log k)}$ for $k=O(n)$ and $2^{\Omega(n^2\log n)}$ for $k=\omega(n)$. Besides, for the lower bound state complexity for determinization construction from NSA to deterministic Streett (transition) automata (DS(T)A) or DP(T)A, a result was given in \cite{Mic88,Lod99} with $2^{\Omega(n\log n)}$ states. Later, Yan \cite{Yan06} obtained the same result via full automata technique. Since then, the lower bound state complexity $2^{\Omega(n\log n)}$ for determinization construction  from NSA to DS(T)A or DP(T)A has never improved. There is a gap between the upper and lower bounds state complexity for determinization construction from NSA to DR(T)A, DS(T)A, or DP(T)A.
Therefore, it is interesting to make the state complexity for Streett determinization construction tight or tighter.


In this paper, we reconstruct {$\mu$-Safra trees} as \emph{H-Safra trees for Streett determinization} by changing the name on each node of the tree. As a consequence, an improved construction of DRTA is obtained with state complexity being $n^{5n}(n!)^{n}$ for $k=\omega(n)$, and $n^{5n}k^{nk}$ for $k=O(n)$.
Then, LIR-H-Safra trees for Streett determinization are presented by adding \emph{later introduction records}, which records the generation order of each node, to H-Safra trees. Based on LIR-H-Safra trees, an improved construction of DPTA is obtained with state complexity being $3(n(n+1)-1)!(n!)^{n+1}$ for $k=\omega(n)$, and $3(n(k+1)-1)!n!k^{nk}$ for $k=O(n)$.
We prove the lower bound state complexity for determinization construction from NSA to DR(T)A by the language game namely $L$-game \cite{CZ09} which matches the state complexity of the proposed determinization construction by H-Safra trees.
Moreover, an improved lower bound state complexity $2^{\Omega(n^2 \log n)}$ for $k=\omega(n)$ and $2^{\Omega(nk \log nk)}$ for $k=O(n)$ for determiniztion construction from NSA to DP(T)A is proposed based on $L$-game. It is the same as the determinization construction by LIR-H-Safra trees in the exponent.

The rest of the paper is organized as follows.
The next section briefly introduces automata over infinite words.
In Section \ref{Sec:history}, Cai and Zhang's NSA-to-DRA determinization based on {$\mu$-Safra trees} is revisited.
Our new data structures, {H-Safra trees for Streett determinization} and {LIR-H-Safra trees for Streett determinization}, are presented in Section~\ref{sec:datastructure}.
In the sequel, the improved NSA-to-DRTA and NSA-to-DPTA determinization constructions are presented in Section~\ref{sec:determinization}.
{Section ~\ref{sec:lower bound} studies the lower bound of the determinization construction.}

\section{Automata}\label{Sec:automata}
Let $\Sigma$ be a finite set of symbols called an alphabet.
 An infinite word $\alpha$ is an infinite
sequence of symbols from $\Sigma$. $\Sigma^\omega$ is the set of all
infinite words over $\Sigma$.  We present $\alpha$ as a function
$\alpha:\mathbb{N}\rightarrow \Sigma$, where $\mathbb{N}$ 
is the set of non-negative integers. Thus, $\alpha(i)$ denotes the
letter appearing at the $i^{th}$ position of the word. In general,
$\mathsf{Inf}(\alpha)$ denotes the set of symbols from $\Sigma$
which occur infinitely often in $\alpha$. Formally,
$\mathsf{Inf}(\alpha)=\{\sigma\in \Sigma\mid\exists^\omega n\in
\mathbb{N}:\alpha(n)=\sigma\}$. Note that $\exists^\omega n \in
\mathbb{N}$ means that there exist infinitely many $n$ in $\mathbb{N}$.

\begin{Def}[Automaton]\rm
An automaton over $\Sigma$ is a tuple $A=(\Sigma, Q,\delta,
Q_0,\lambda)$, where $Q$ is a 
non-empty, finite set of states, $Q_0\subseteq Q$ is a set of initial
states, ${\delta} \subseteq
Q\times \Sigma\times Q$ 
is a transition relation, and $\lambda$ is an acceptance condition.
\end{Def}

A run $\rho$ of an automaton $A$ on an infinite word $\alpha$ is an
infinite sequence $\rho:\mathbb{N}\rightarrow Q$ such that
$\rho(0)\in Q_0$ and for all $i\in \mathbb{N}$,
$(\rho(i),\alpha(i),\rho(i+1))\in {\delta}$.  $A$ is said to be
deterministic if $Q_0$ is a singleton, and for any $(q,\sigma,q')\in$
$\delta$, there exists no $(q,\sigma,q'')\in \delta$ such that
$q''\not=q'$, and nondeterministic otherwise. Similar to infinite
words, $\mathsf{Inf}(\rho)$ denotes the set of states from $Q$ which
occur infinitely often in $\rho$. Formally,
$\mathsf{Inf}(\rho)=\{q\mid\exists^\omega n\in
\mathbb{N}:\rho(n)=q\}$.

Several acceptance conditions are studied in literature. We
present three of them here:

\begin{itemize}

\item {Streett, where $\lambda=\{\langle G_1,B_1\rangle,\langle G_2,B_2\rangle,\ldots,\langle G_k,B_k\rangle\}$ with
$G_i$, $B_i\subseteq Q$.  $\rho$ is accepted iff for all $1\leq i\leq k$, we
have that $\mathsf{Inf}(\rho)\cap G_i\neq\emptyset$ or
$\mathsf{Inf}(\rho)\cap B_i=\emptyset$.}

\item Rabin, where $\lambda=\{\langle A_1,R_1\rangle, \langle A_2, R_2\rangle, \ldots,
\langle A_{k},R_k\rangle\}$ with $A_i$, $R_i\subseteq Q$.
$\rho$ is accepted iff for some $1\leq i\leq k$, we have that
$\mathsf{Inf}(\rho)\cap A_i\not=\emptyset$ and
$\mathsf{Inf}(\rho)\cap R_i=\emptyset$.

{\item Parity, where $\lambda=\{\lambda_1,\lambda_2,\ldots,\lambda_{2k}\}$ with $\lambda_1\cup\lambda_2\cup\ldots\cup\lambda_{2k}=Q$. $\rho$ is accepted iff the minimal index $i$ for which $\mathsf{Inf}(\rho)\cap\lambda_i\neq\emptyset$ is even.}
\end{itemize}

An automaton accepts a word if it has an accepting run on it. The
accepted language of an automaton $A$, denoted by $L(A)$, is the set
of words that $A$ accepts.

We denote the different types of automata by three letter acronyms
in $\{D,N\}\times \{S,R,P\}\times \{A\}$. The first letter stands
for the branching mode of the automaton (deterministic or
nondeterministic); the second letter stands for the acceptance
condition type (Streett, Rabin, or parity); and the third letter
indicates automata. While acceptance
condition of an ordinary automaton is defined on states, the
acceptance condition of a transition automaton is defined on
transitions of the automaton. {Accordingly, with respect to each type
of ordinary automata, we also have its transition version.}

\section{Determinization via $\mu$-Safra Trees for Streett}\label{Sec:history}

This section revisits the determinization construction via {$\mu$-Safra trees for Streett}
\cite{CZ12}. {{For any positive integer $m\in\mathbb{N}$, we use $[m]$ to denote the set $\{1,2,\ldots,m\}$.}}

\subsection{$\mu$-Safra Trees for Streett Determinization}
$\mu$-Safra trees for Streett determinization, presented by Cai and Zhang in 2012 \cite{CZ12}, are obtained from Streett Safra trees \cite{Safra92}. A $\mu$-Safra tree for Streett determinization is a labelled ordered tree. A tree is \emph{ordered} just if the nodes are partially ordered by \emph{older-than} relation.
Compared with Streett Safra trees, the characteristic of $\mu$-Safra trees for Streett determinization is a \emph{batch-mode naming scheme $M_b$} for nodes.

For an ordered tree, a leaf corresponds to a \emph{left spine}. A left spine is a maximal path $\tau_1,\tau_2,\ldots,\tau_m$ such that $\tau_m$ is a leaf, for any $i\in\{2,\ldots,m\}$, $\tau_i$ is the left-most child of $\tau_{i-1}$, and $\tau_1$, called the head of the left spine, is not a left-most child of its parent \cite{CZ12}. We arrange all left spines with consecutive integers starting from $1$ as names of left spines. Each node is on exactly one left spine.
For the sibling nodes, the name of the left spine, which contains the left-most sibling, is smaller than the others.
With this basis, every node can be named uniquely. Nodes in a left spine named $ls$, from the head downwards, are assigned continuously increasing names, starting from $ls.1$.

\begin{Rule} [Batch-mode naming scheme $M_b$]\label{rename:old} \rm
If a node $\tau$ belongs to the left spine named $ls$, and $\tau$ is the $i$-th node in $ls$, the name of $\tau$ is $ls.i$, i.e. $M_b(\tau)=ls.i$ \cite{CZ12}.
\end{Rule}

\begin{Def} [$Cover$ and $Mini$ \cite{CZ12,CZ11(1)}] \rm \label{Cover-Mini}
For a NSA $S = (\Sigma, Q,$ $Q_0,\delta, \lambda)$ with $|Q|=n$ and $k$ Streett pairs $\lambda = \{\langle G_1,B_1\rangle,$ $\langle G_2,B_2\rangle,\ldots,\langle G_k,B_k\rangle\}$. Let  $\beta$  be a subset of $[k]$, and $G_{\beta} = \bigcup_{i\in\beta}G_i$, where $G_i$  is the first element of the $i$-th Streett pair $\langle G_i,B_i\rangle$. Then, $Cover$ maps $2^{[k]}$ to $2^{[k]}$ such that
$$Cover(\beta)=\{j\in [k] \mid G_j \subseteq G_{\beta}\}$$
$Mini$ also maps $2^{[k]}$ to $2^{[k]}$ such that $j\in Mini(\beta)$ if, and only if, $j\in [k]\backslash Cover(\beta)$ and
\begin{equation}
\forall j'\in [k]\backslash Cover(\beta),[j'\neq j \rightarrow (G_{j'}\cup G_{\beta}\not\subset G_{j}\cup G_{\beta})], \label{con1}
\end{equation}
\begin{equation}
\forall j'\in [k]\backslash Cover(\beta),[j'<j \rightarrow (G_{j'}\cup G_{\beta} \neq G_{j}\cup G_{\beta})]. \label{con2}
\end{equation}
\end{Def}

\begin{Exm} \rm
For a NSA with $n=3$, $k=4$, $Q=\{q_0,q_1,q_2\}$, and the first elements of the four Streett pairs are
$G_1=\{q_0,q_1\}$, $G_2=\{q_0\}$, $G_3=\{q_1,q_2\}$, and $G_4=\{q_2\}$.
Let $\beta = \{3\}$. We have $G_{\beta}=G_3=\{q_1,q_2\}$. Obviously, $G_3\subseteq G_{\beta}$ and $G_4\subseteq G_{\beta}$, which infers to $Cover(\beta)=\{3,4\}$.

Further, we have $[k]\backslash Cover(\beta)=\{1,2\}$. 
For $j=1$, $j'=2$, we have $G_{j'}\cup G_{\beta} = \{q_0,q_1,q_2\}$ and $G_{j}\cup G_{\beta} = \{q_0,q_1,q_2\}$, which satisfies Conditions (\ref{con1}) and (\ref{con2}). Thus, $1\in Mini(\beta)$.
For $j=2$, $j'=1$, we also have $G_{j'}\cup G_{\beta} = \{q_0,q_1,q_2\}$ and $G_{j}\cup G_{\beta} = \{q_0,q_1,q_2\}$. Obviously, Condition (\ref{con2}) is violated since $(j'=1)<(j=2)$. Thus, $2\notin Mini(\beta)$.
As a result, $Mini(\beta)=\{1\}$.
\end{Exm}

\begin{Def} [$\mu$-Safra tree for Streett determinization \cite{CZ12}] \rm
Fix a NSA $S = (\Sigma, Q, Q_0,$ $\delta, \lambda)$ with $|Q|=n$ and $k$ Streett pairs $\lambda = \{\langle G_1,B_1\rangle,\langle G_2,B_2\rangle,\ldots,\langle G_k,B_k\rangle\}$. A \emph{$\mu$-Safra tree for Streett determinization} of the NSA $S$ is a labeled ordered tree $\langle T_o, V, l, h, M_b, E, F,$ $stor \rangle$, where $T_o$ is an ordered tree, and
\begin{itemize}
  \item $V$ is the set of all nodes in $T_o$.
  \item $l$: $V\rightarrow 2^Q$ is a state label of nodes with subsets of $Q$. The label of every node is equal to the union of its sons. The labels of two siblings are disjoint.
  \item $h$: $V\rightarrow 2^{[k]}$ is an index label, which annotates every node with a set of indices from $[k]$. The root is annotated by $[k]$. The annotation of every node is contained in that of its parent and it misses at most one element from the annotation of the parent. Every node that is not a leaf has at least one son with strictly smaller annotation. In addition, each leaf $\tau_l$ satisfies $h(\tau_l)=\emptyset$ or $Mini([k]-h(\tau_l))=\emptyset$, where $Mini$ is defined in Definition \ref{Cover-Mini} for determining the index labels of nodes.
  \item $M_b$: $V\rightarrow [n].[\mu+1]$, where $\mu=\min(n,k)$, assigns each node a unique name by the batch-mode naming scheme.
  \item $E,F\subseteq V$ are two disjoint subsets of $V$. They are used to define the Rabin acceptance condition.
  \item $stor$ is an additional \emph{structural ordering} on nodes. For every non-root node $\tau$, let $j(\tau)=\max\{(h(\tau_p)\cup\{0\})-h(\tau)\}$ where $\tau_p$ is the parent of $\tau$. $stor$ means that for any two siblings $\tau$ and $\tau'$, $\tau'$ is placed to the right of $\tau$ if, and only if, $j(\tau)>j(\tau')$, or $j(\tau)=j(\tau')$ and $\tau$ is older than $\tau'$.
\end{itemize}
\end{Def}

The following lemma has been proved in \cite{CZ12}.

\begin{Lem} \rm  \label{TreeSize1}
For a $\mu$-Safra tree for Streett determinization of a NSA with $n$ states and $k$ Streett pairs, there are at most $n$ left spines, and each left spine has at most $\mu+1$ nodes, where $\mu=\min(n,k)$. Therefore, $[n].[\mu+1]$ node names  are sufficient \cite{CZ12}.
\end{Lem}

Accordingly, Lemma \ref{TreeSize2} is easily obtained.

\begin{Lem} \rm \label{TreeSize2}
The number of nodes in a $\mu$-Safra tree for Streett determinization is at most $n(\mu +1)$ \cite{CZ12}.
\end{Lem}

Fig.~\ref{fig:ht1} illustrates a $\mu$-Safra tree for Streett determinization of a NSA with $5$ states, namely, $a$, $b$, $c$, $d$ and $e$. This $\mu$-Safra tree  contains $12$ nodes. The state sets shown in nodes are state labels. The batch-mode names and index labels of nodes are given in red and blue, respectively. There are four left spines, i.e. $\{1.1,1.2,1.3,1.4\}$, $\{2.1,2.2,2.3\}$, $\{3.1,3.2\}$ and $\{4.1,4.2,4.3\}$.

\begin{figure}[htp]
\centerline{\includegraphics[height=4.5cm]{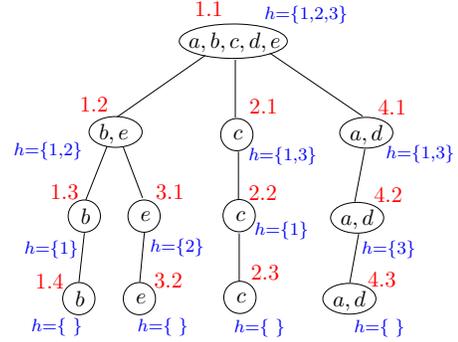}}\caption{A $\mu$-Safra tree for Streett determinization} \label{fig:ht1}
\end{figure}

Along a sequence of $\mu$-Safra tree for Streett determinization transformations, there may exist some node whose name is changed. For instance, when a node moves into another left spine, the node should be renamed.  The renaming scheme is stated by Rule~\ref{rename:old} \cite{CZ12}.

\begin{Rule} [Batch-mode renaming scheme]\label{rename:old} \rm
  When a left spine is created, nodes in the left spine are assigned names from an unused name bucket. When a left spine is removed, the name bucket of the left spine is recycled.
   When a left spine $ls$ is grafted into another left spine $ls'$, the name bucket of $ls$ is recycled and nodes on $ls$ are renamed as if they were on $ls'$, originally.
\end{Rule}

In the transformations, the index labels $h$ of the new created nodes also need to be defined. The index label $h$ of a node $\tau$ is a subset of the indices set of all Streett pairs. We will check whether all states in $l(\tau)$ visits the first elements $G$s of these Streett pairs one by one. But there may exist a situation that some $G$ of a Streett pair $\langle G,B\rangle$ is contained by another $G'$. If $G'$ has been checked, it is redundant to further check $G$. In order to reduce unnecessary inspections, the functions $Mini$, which decreases the combination of index labels $h$, will be utilized in the determinization construction. In \cite{CZ12}, it has been proved that using $Mini$ to select the index labels of the children is sound and complete.

\subsection{Construction of $\mu$-Safra Trees for Streett Determinization}
Fix a NSA $S=(\Sigma,Q,Q_0,\delta,\langle G,B\rangle_{[k]})$. The initial $\mu$-Safra tree for Streett determinization of $S$ is a single-branch (only a left spine) labelled tree $T_I$. Every node is named by the batch-mode naming scheme. For each node $\tau$, $l(\tau)=Q_0$ and $h(\tau)=h(\tau_p)-\max\{Mini([k]-h(\tau_p))\}$. Specially, for the root $\tau_r$, $h(\tau_r)=[k]$, and the leaf $\tau_l$ satisfies $h(\tau_l)=\emptyset$ or $Mini([k]-h(\tau_l))=\emptyset$. Set $E=\emptyset$ and $F=\emptyset$. Given a $\mu$-Safra tree $T_{\mu}$ for Streett determinization of $S$ and $\sigma\in\Sigma$, we construct a new $\mu$-Safra tree $\hat{T_{\mu}}$ for Streett determinization, called the $\sigma$-$successor$ of $T_{\mu}$, in six steps as follows.

\begin{enumerate}
\item \textbf{Update}: Set $E$ and $F$ to empty sets and replace the state label of every node $\tau$ in $T_{\mu}$ by $\bigcup_{q \in l(\tau)} \delta(q,\sigma)$. Call the resultant labelled tree $T_{\mu_1}$.

\item \textbf{Create siblings}: Apply the following transformations to non-leaf nodes of $T_{\mu_1}$. Let $\tau$ be a node with $m$ children $\tau_1,\ldots,\tau_m$. Sequentially consider the following cases for each $i\in[1..m]$ from $1$ to $m$.

    \begin{enumerate}[a)]
      \item If $l(\tau_i)\cap G_{j(\tau_i)}\neq\emptyset$, add a child $\tau'$ to $\tau$ with $l(\tau')=l(\tau_i)\cap G_{j(\tau_i)}$ and $h(\tau')=h(\tau)-\max\{[0..j(\tau_i))\cap(\{0\}\cup Mini([k]-h(\tau)))\}$, and remove the states in $l(\tau_i)\cap G_{j(\tau_i)}$ from $\tau_i$ as well as all its descendants.

      \item If $l(\tau_i)\cap G_{j(\tau_i)}=\emptyset$ and $l(\tau_i)\cap B_{j(\tau_i)}\neq\emptyset$, add a child $\tau'$ to $\tau$ with $l(\tau')=l(\tau_i)\cap B_{j(\tau_i)}$ and $h(\tau')=h(\tau_i)$, and remove the states in $l(\tau_i)\cap B_{j(\tau_i)}$ from $\tau_i$ as well as all its descendants.
    \end{enumerate}
    Call the resultant labelled tree $T_{\mu_2}$.

\item \textbf{Horizontal merge}: For any two siblings $\tau$ and $\tau'$ in $T_{\mu_2}$ and any state $q\in l(\tau_i)\cap l(\tau_{i'})$, if $j(\tau) <j(\tau')$, or $j(\tau) = j(\tau')$ and $\tau$ is older than $\tau'$, then remove $q$ from $\tau'$ and all its descendants. Remove nodes with empty state label and add their names, if defined, to $E$. Call the resultant labelled tree $T_{\mu_3}$.

\item \textbf{Vertical merge}: For each non-leaf $\tau$ in $T_{\mu_3}$, if all children are annotated by $h(\tau)$, then remove all the children and their descendants. Add the name of $\tau$ to $F$. Call the resultant labelled tree $T_{\mu_4}$.

\item \textbf{Rename}: Rename nodes whose names are defined in $T_{\mu_4}$ according to Rule \ref{rename:old} and add nodes that are renamed to $E$, which results in $T_{\mu_5}$.

\item \textbf{Create children}: Repeat the following procedure until no new nodes can be added: For each leaf $\tau$ in $T_{\mu_5}$ such that $h(\tau)\neq\emptyset$ and $Mini([k]-h(\tau))\neq\emptyset$, add to $\tau$ a new child $\tau'$. Set $l(\tau')=l(\tau)$, $h(\tau')=h(\tau)-\max\{Mini([k]-h(\tau))\}$. Then name nodes whose names are undefined according to the batch-mode naming scheme. The resultant labelled tree is denoted as $\hat{T_{\mu}}$.
\end{enumerate}

$\hat{T_{\mu}}$ is a $\mu$-Safra tree for Streett determinization.

Thus, given a NSA $S=(\Sigma,Q,Q_0,\delta,\langle G,B \rangle_{[k]})$, by applying the above six-step procedure recursively until no new $\mu$-Safra trees can be created, an associated DRA $DR=(\Sigma,Q_{DR},T_{\mu_I},$\\$\delta_{DR},\lambda_{DR})$ can be constructed. Here, $Q_{DR}$ is the set of $\mu$-Safra trees for Streett determinization of $S$, $T_{\mu_I}$ is the initial $\mu$-Safra tree for Streett determinization, $\delta_{DR}$ is the $\mu$-Safra-tree-Streett transition relation (i.e. $T_{\mu}\xrightarrow{\sigma} \hat{T_{\mu}}$ whenever $\hat{T_{\mu}}$ is the $\sigma$-successor of $T_{\mu}$), and $\lambda_{DR}=\{(A_{\tau_1},R_{\tau_1}),\ldots,(A_{\tau_k},R_{\tau_k})\}$ (where $k\geq 1$) is the Rabin acceptance condition. {For each $i$, the node $\tau_i$ is given by its name, $A_{\tau_i}$ is the set of $\mu$-Safra trees for Streett determinization  (node $\tau_i$ belongs to $F$ ), and $R_{\tau_i}$ the set of $\mu$-Safra trees for Streett determinization (node $\tau_i$ belongs to $E$).}

Given an input $\omega$-word $\alpha : \omega\rightarrow\Sigma$, we call the sequence $\Pi = T_{\mu_0}T_{\mu_1}T_{\mu_2}T_{\mu_3}\ldots$ of $\mu$-Safra trees for Streett determinization such that $T_{\mu_0} = T_{\mu_I}$, and for all $i\in\omega$, $T_{\mu_{i+1}}$ is the $\alpha(i)$-successor of $T_{\mu_i}$, the \emph{$\mu$-Safra Streett trace} of the NSA $S$ over $\alpha$. We view the $\mu$-Safra Streett trace of $S$ over $\alpha$ as the run of the DRA $DR$ over $\alpha$. Then we say that $\alpha$ is \emph{accepted by the DRA} if there exists $i\in\{1,\ldots,k\}$ such that $\mathsf{Inf}(\Pi)\cap A_{\tau_i}\neq\emptyset$ and $\mathsf{Inf}(\Pi)\cap R_{\tau_i}=\emptyset$.

\begin{Thm} [Cai and Zhang \cite{CZ12,CZ13}]
Given a NSA $S$ with $n$ states and $k$ Streett pairs, a DRA with $n^{7n}(n!)^{n+1}$ states, $O(n^2)$ Rabin pairs for $k=\omega(n)$, and $n^{5n}k^{n(k+2)}n!$ states, $O(nk)$ Rabin pairs for $k=O(n)$ can be constructed that recognizes the language $L(S)$.
\end{Thm}

By deleting the two sets $E$ and $F$ of each $\mu$-Safra tree in the Streett determinization and recording  the accepting and rejecting nodes throughout each transition, a DRTA can be constructed.

\begin{Cor}
Given a NSA $S$ with $n$ states and $k$ Streett pairs, a DRTA with $n^{5n}(n!)^{n+1}$ states, $O(n^2)$ Rabin pairs for $k=\omega(n)$, and $n^{5n}k^{nk}n!$ states, $O(nk)$ Rabin pairs for $k=O(n)$ can be constructed that recognizes the language $L(S)$.
\end{Cor}

\section{H-Safra Trees and LIR-H-Safra Trees for Streett Determinization}\label{sec:datastructure}
This section presents two new data structures, called H-Safra trees and LIR-H-Safra trees for Streett determinization.
\subsection{H-Safra Trees for Streett Determinization}
{As for B\"{u}chi determinization, Schewe proposes a tight construction via \emph{history trees} which results in an equivalent DRTA \cite{Sven09}. In Schewe's construction, instead of explicit names,  nodes are implicitly named. This leads to a reduction of state complexity. 
With this motivation, we put forward a new data structure namely \emph{H-Safra trees} for Streett determinization. Compared with $\mu$-Safra trees for Streett determinization, the only difference is the naming scheme of nodes.

For a \emph{structural ordered tree with state and index labels} (i.e. a $\mu$-Safra tree for Streett determinization without names, $E$ and $F$), denoted by $T_{si}$ (Fig.~\ref{fig:ht2} is an example), we give a new naming scheme depending only on the index label $h$ of nodes, which is expressed by Rule~\ref{name:new}.

\begin{figure}[htp]
\centerline{\includegraphics[height=4.5cm]{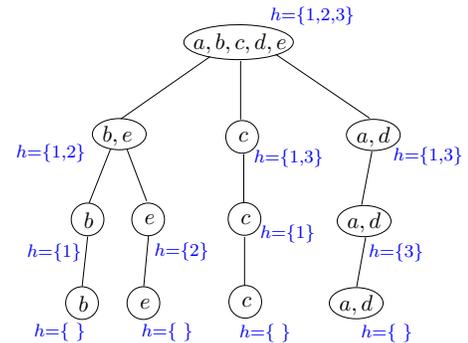}}\caption{A structural ordered tree with state and index labels} \label{fig:ht2}
\end{figure}

\begin{Rule} [Naming scheme $M_n$]\label{name:new} \rm \
\begin{itemize}
  \item For the root $\tau_r$, $M_n(\tau_r)=\epsilon$;
  \item for each node $\tau$ in the second level, $M_n(\tau)=j(\tau)^{i+1}$ where $i=|\{\tau'|\tau'$ is the left sibling of $\tau$, and $j(\tau')=j(\tau)\}|$;
  \item for any other node $\tau$, $M_n(\tau)=M_n(\tau_p).j(\tau)^{i+1}$.
\end{itemize}
\end{Rule}

Utilizing the new naming scheme, we can get a \emph{H-Safra tree for Streett determinization}.
\begin{Def} [H-Safra trees for Streett determinization] \rm
A \emph{H-Safra tree for Streett determinization} of a given NSA $S=(\Sigma,Q,Q_0,\delta,\lambda)$ with $n$ states and $k$ Streett pairs is a pair $\langle T_{si},M_n \rangle$ where $T_{si}$ is a structural ordered tree with state and index labels of $S$, and $M_n$ is the new naming scheme.
\end{Def}

Fig.~\ref{fig:ht3} is a H-Safra tree for Streett determinization obtained from Fig.~\ref{fig:ht2} by using the new naming scheme. Here, the names of nodes are given in red. For the node $\tau$ with $l(\tau)=\{a,d\}$ and $h(\tau)=\{1,3\}$ ($j(\tau)=2$), it belongs to the second level nodes, and there exists a left sibling $\tau'$ such that $j(\tau')=j(\tau)=2$. Thus the name of $\tau$ is $2^2$.

\begin{figure}[htp]
\centerline{\includegraphics[height=5cm]{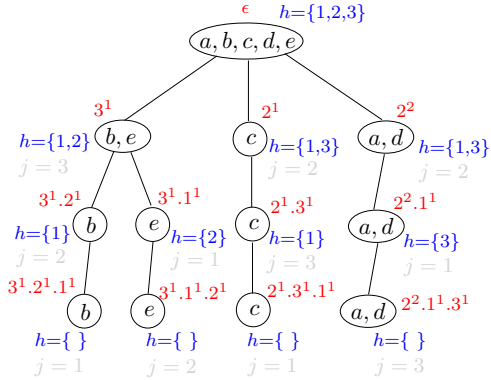}}\caption{A H-Safra tree} \label{fig:ht3}
\end{figure}

Obviously, each node in a structural ordered tree with state and index labels can be uniquely named.

The new naming scheme is the core of our determinization construction. Given a NSA, H-Safra trees for Streett determinization will be taken as the states of the final DRTA. By the naming scheme, once the index label $h$ of each node is fixed, the name is also determined, which makes the state complexity decrease.

\begin{Lem} \rm \label{complex:HST}
The number of H-Safra trees for Streett determinization of a given NSA is equal to the number of structural ordered trees with state and index labels, i.e. $\mu$-Safra trees for Streett determinization without names, $E$, and $F$, occurring in the determinization construction.
\end{Lem}

\begin{proof}
By the naming scheme $M_n$, for each node $\tau$ occurring in a structural ordered tree with state and index labels,
a unique name $M_n(\tau)$ is assigned to $\tau$. $M_n(\tau)$ depends on the index label and the position of $\tau$ in the tree. Thus, the number of H-Safra trees for Streett determinization of a NSA is equal to the number of structural ordered trees with state and index labels.
\end{proof}

\subsection{LIR-H-Safra Trees for Streett Determinization}
In order to transform a NSA to a DPTA, we need a dynamic node identification scheme that captures the order in which the nodes are created when constructing the $\sigma$-successors. Consequently, the state complexity of the DPTA transform will increase. Similar to the constructions of Schewe from NBA to DPA \cite{Sven09} and from NPA to DPA \cite{ScheweV14}, the data structure we shall use is H-Safra trees for Streett determinization with \emph{later introduction record} (LIR), called \emph{LIR-H-Safra trees for Streett determinization}. A LIR is a sequence of nodes in the H-Safra tree for Streett determinization according to the order the nodes are generated.

\begin{Def} [LIR-H-Safra trees for Streett determinization] \em
Given a NSA $S=(\Sigma,$ $Q,Q_0,\delta,\lambda)$ with $n$ states and $k$ Streett pairs, a LIR-H-Safra tree for Streett determinization is a pair $\langle H,LIR\rangle$ where $H$ is a H-Safra tree for Streett determinization and $LIR$ stores the order in which the nodes of $H$ are created.
\end{Def}

\begin{figure}[htp]
\centerline{\includegraphics[height=5.7cm]{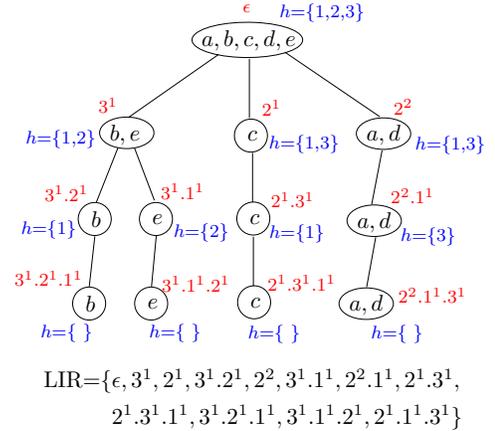}}\caption{A LIR-H-Safra tree} \label{fig:LHSTS}
\end{figure}

Fig.~\ref{fig:LHSTS} is a LIR-H-Safra tree for Streett determinization. The LIR contains all nodes of the tree such that each node appears after its left siblings. Every node in LIR is represented  by its name for simplicity.

As for each node $\tau$ of a given LIR-H-Safra tree for Streett determinization, we introduce an extra notation $p(\tau)$ to denote the position of $\tau$ in the LIR.

\section{Determinization via H-Safra Trees and LIR-H-Safra Trees for Streett}\label{sec:determinization}
This section presents a NSA-to-DRTA determinization transform via H-Safra trees and a NSA-to-DPTA determinization transform via LIR-H-Safra trees.

\subsection{Construction of H-Safra Trees for Streett Determinization}\label{sec:NSA2DRTA}
Fix a NSA $S=(\Sigma,Q,Q_0,\delta,\langle G,B\rangle_{[k]})$. The \emph{initial H-Safra tree for Streett determinization} of $S$ is a single-branch labelled tree $H_I$.  For each node $\tau$ of $H_I$, the state label $l(\tau)=Q_0$ and index label $h(\tau)=h(\tau_p)-\max\{Mini([k]-h(\tau_p))\}$. Specially, for the root $\tau_r$, $h(\tau_r)=[k]$, and the leaf $\tau_l$ satisfies $h(\tau_l)=\emptyset$ or $Mini([k]-h(\tau_l))=\emptyset$. Every node in $H_I$ is named by the new naming scheme.

Given a H-Safra tree $H$ for Streett determinization of $S$ and $\sigma\in\Sigma$, we construct a new H-Safra tree $\hat{H}$ for Streett determinization, called the $\sigma$-\emph{successor} of $H$, and the \emph{signatures} $\sig_{acc}$ and $\sig_{rej}$ of the transition, in six steps as follows.

We intuitively illustrate the six steps of construction by an example. Fig.~\ref{fig:nsa} shows all transitions for an input letter $\sigma$ from the states in the H-Safra tree for Streett determinization in Fig.~\ref{fig:ht3}.

\begin{figure}[htp]
\centerline{\includegraphics[height=2.7cm]{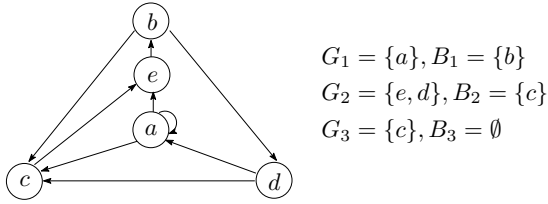}}\caption{Relevant fragment of a Streett automaton} \label{fig:nsa}
\end{figure}

\paragraph*{\rm \textbf{Step 1: Update}}
Replace the state label of every node $\tau$ in $H$ by $\bigcup_{q \in l(\tau)} \delta(q,\sigma)$. Call the resultant labelled tree $H_1$.

Let $H$ be the H-Safra tree for Streett determinization in Fig.\ref{fig:ht3} for the NSA whose transition is depicted in Fig.\ref{fig:nsa}. Fig.\ref{fig:step1} shows the tree structure $H_1$ resulting from $H$ after Step 1 of the construction procedure. Compared with $H$, state labels of all nodes in $H_1$ are updated.

\begin{figure}[htp]
\centerline{\includegraphics[height=4.3cm]{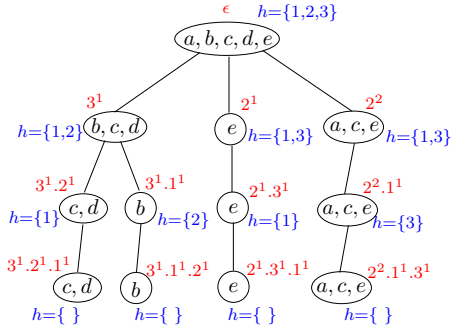}}\caption{Step 1 of the construction procedure} \label{fig:step1}
\end{figure}

\paragraph*{\rm \textbf{Step 2: Create siblings}}
Apply the following transformations to non-leaf nodes of $H_1$  from the root. Let $\tau$ be a node with $m$ children $\tau_1,\ldots,\tau_m$. Sequentially consider the following two cases for each $i\in [1..m]$ from $1$ to $m$:
\begin{enumerate}[a)]
      \item If $l(\tau_i)\cap G_{j(\tau_i)}\neq\emptyset$, add a youngest child $\tau'$ to $\tau$ with $l(\tau')=l(\tau_i)\cap G_{j(\tau_i)}$ and $h(\tau')=h(\tau)-\max\{[0..j(\tau_i))\cap(\{0\}\cup Mini([k]-h(\tau)))\}$, and remove the states in $l(\tau_i)\cap G_{j(\tau_i)}$ from $\tau_i$ and all its descendants; then
      \item if $l(\tau_i)\cap B_{j(\tau_i)}\neq\emptyset$, add a youngest child $\tau'$ to $\tau$ with $l(\tau')=l(\tau_i)\cap B_{j(\tau_i)}$ and $h(\tau')=h(\tau_i)$, and remove the states in $l(\tau_i)\cap B_{j(\tau_i)}$ from $\tau_i$ and all the descendants.
\end{enumerate}
Note that the names of the new created nodes are not defined currently. Then rearrange sibling nodes by the structural ordering from the second level to the last level.

 We use a simple example illustrated in Fig.\ref{fig:stor} to show how the sibling nodes are rearranged. For the siblings $\tau_1, \tau_2, \tau_3$ and $\tau_4$ in Fig.\ref{fig:stor} (a), we have $j(\tau_1)=2$, $j(\tau_2)=3$, $j(\tau_3)=1$, and $j(\tau_4)=2$. We rearrange the siblings according to the value of $j$  from the largest to the smallest. As for $\tau_1$ and $\tau_4$ with $j(\tau_1)=j(\tau_4)$, $\tau_4$ is younger than $\tau_1$, since the later the node is generated, the younger it is. It indicates that the relative order of nodes with the same $j$ will not change. The resultant tree after  structural ordering is shown in  Fig.\ref{fig:stor} (b). Compared with Fig.\ref{fig:stor} (a), the positions of $\tau_1$ and $\tau_2$, and $\tau_3$ and $\tau_4$ are swapped, respectively.

\begin{figure}[htp]
\centerline{\includegraphics[width=5.25cm]{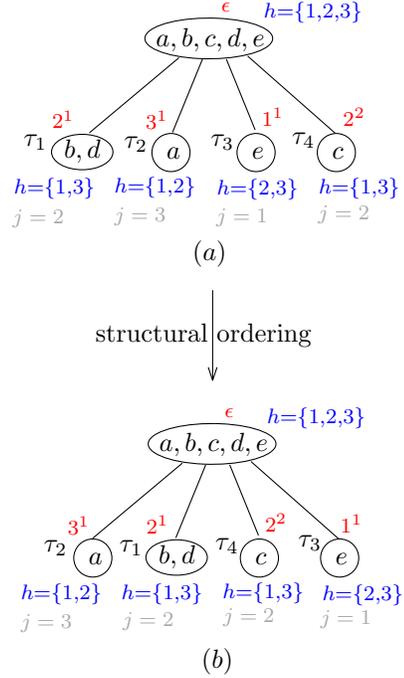}}\caption{Structural ordering} \label{fig:stor}
\end{figure}

After Step 2, the resultant labelled tree, called $H_2$, is shown in Fig.\ref{fig:step2}.
The nodes without names are new created in this step, and every node observes the structural ordering. The state labels of the nodes in grey will be deleted in Step 3.

\begin{figure}[htp]
\centerline{\includegraphics[height=4.3cm]{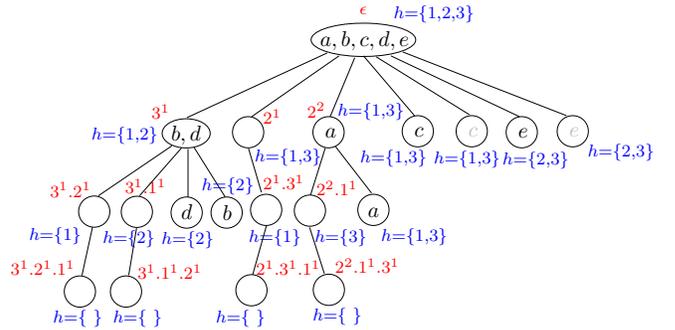}}\caption{Step 2 of the construction procedure} \label{fig:step2}
\end{figure}

\paragraph*{\rm \textbf{Step 3: Horizontal merge}}
For each node $\tau$ in $H_2$ starting from the root, and every state $q\in l(\tau)$, if $q$ also occurs in the state label of a sibling $\tau'$ of $\tau$ such that $j(\tau')<j(\tau)$, or $j(\tau')=j(\tau)$ and $\tau'$ is older than $\tau$, then remove $q$ from $\tau$ as well as all its descendants. Afterward, for any node $\tau$,  remove $\tau$ if $l(\tau)=\emptyset$. A removed node whose name is defined is called \emph{rejecting}.

Let $H_3$ be the resultant tree. Next, we define $\sig_{rej}=\{\tau \mid \tau$ is the rejecting node occurring in the current tree$\}$, called the \emph{rejecting signature} of the $\delta$-successor / transition being defined. The resulting tree is depicted in Fig.~\ref{fig:step3} with $\sig_{rej}=\{2^1,$ $3^1.2^1,3^1.1^1,2^1.3^1,2^2.1^1,$ $3^1.2^1.1^1,3^1.1^1.2^1,2^1.3^1.1^1,2^2.1^1.3^1\}$. In the resultant tree, the state labels of the siblings are pairwise disjoint and there exists no empty node. {Nevertheless, there may exist a node which is equal to each of its children in index label.}

\begin{figure}[htp]
\centerline{\includegraphics[height=3.7cm]{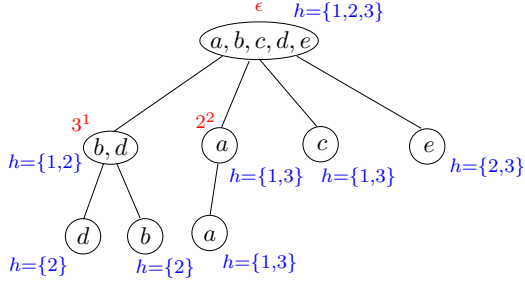}}\caption{Step 3 of the construction procedure} \label{fig:step3}
\end{figure}

\paragraph*{\rm \textbf{Step 4: Vertical merge}}
For each non-leaf $\tau$ in $H_3$ starting from the root, if the index label of each child is equal to $h(\tau)$, then remove all the children of $\tau$ as well as their descendants. The nodes whose descendants have thus been removed are called \emph{accepting}.

Let $H_4$ be the resultant tree. Next define $\sig_{acc}=\{\tau \mid \tau$ is the accepting node occurring in the current tree$\}$, called the \emph{accepting signature} of the $\delta$-successor / transition being defined. The resulting tree is depicted in Fig.~\ref{fig:step4} with $\sig_{acc}=\{2^2\}$. The state labels of the siblings are pairwise disjoint, {and no node  is equal to each of  its children in index label.} The names of nodes might not follow the new naming scheme. The nodes that will be renamed in Step 4 are drawn in red.

\begin{figure}[htp]
\centerline{\includegraphics[height=3.8cm]{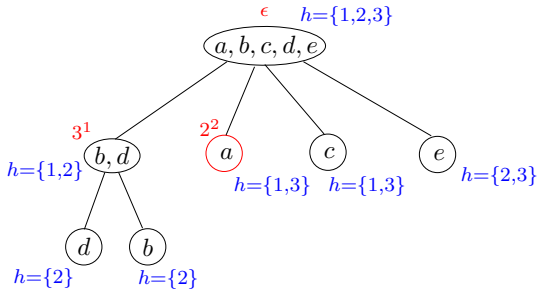}}\caption{Step 4 of the construction procedure} \label{fig:step4}
\end{figure}

\paragraph*{\rm \textbf{Step 5: Rename}}
Rename nodes whose names are defined in $H_4$ starting from the root by applying the naming scheme (Rule \ref{name:new}). {The nodes which should be renamed are also rejecting in this step.} Add these rejecting nodes to $\sig_{rej}$. As for this example, $\sig_{rej}=\{2^1,$ $3^1.2^1,3^1.1^1,2^1.3^1,2^2.1^1,3^1.2^1.1^1,$ $3^1.1^1.2^1,2^1.3^1.1^1,2^2.1^1.3^1,2^2\}$.

Call the resultant labelled tree $H_5$. Fig.~\ref{fig:step5} shows the tree that results from Step 5. All nodes observe the naming scheme. Then the resultant tree will spawn in the next step.

\begin{figure}[htp]
\centerline{\includegraphics[height=3.8cm]{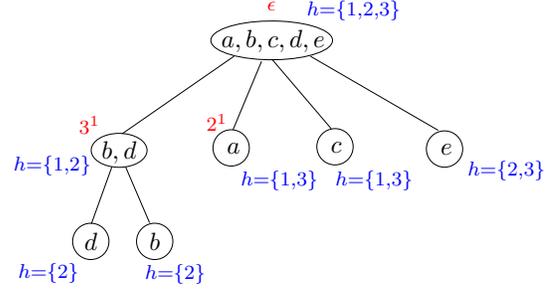}}\caption{Step 5 of the construction procedure} \label{fig:step5}
\end{figure}

\paragraph*{\rm \textbf{Step 6: Create children}}
Repeat the following procedure until no new nodes can be added: For each leaf $\tau$ in $H_5$, if $h(\tau)\neq\emptyset$ and $Mini([k]-h(\tau))\neq\emptyset$, add to $\tau$ a new child $\tau'$. Set $l(\tau')=l(\tau)$ and $h(\tau')=h(\tau)-\{\max(Mini([k]-h(\tau)))\}$. Then define names of the nodes which have not been named by the new naming scheme yet.

The resultant labelled tree is a H-Safra tree for Streett determinization, which we call $\hat{H}$. Note that given $H$ and $\sigma\in\Sigma$, there are a unique $\sigma$-successor $\hat{H}$,   $\sig_{acc}$,  and $\sig_{rej}$. Fig.~\ref{fig:step6} shows $\hat{H}$, called the $\sigma$-\emph{successor} of $H$, obtained through the six steps. Note that states in the resultant DRTA are H-Safra trees for Streett determinization, and the signatures $\sig_{acc}$, $\sig_{rej}$ are part of the transition relation of the DRTA transform.

\begin{figure}[htp]
\centerline{\includegraphics[height=4.8cm]{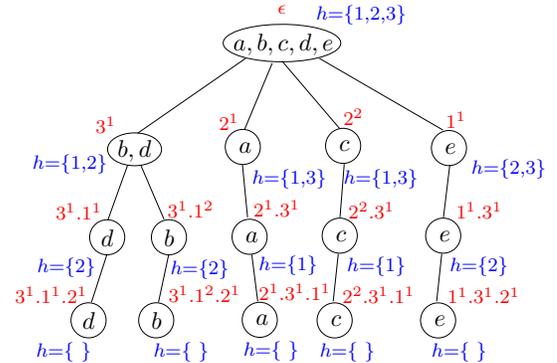}}\caption{Step 6 of the construction procedure} \label{fig:step6}
\end{figure}

Based on the six-step procedure, given a NSA $S=(\Sigma,Q,$ $Q_0,\delta,\langle G,B\rangle_{[k]})$, an equivalent DRTA $RT=(\Sigma,Q_{RT},Q_{RT0},$ $\delta_{RT},\lambda_{RT})$ can be obtained. Here $Q_{RT}$ is the set of H-Safra trees for Streett determinization w.r.t $S$; $Q_{RT0}$ is the initial H-Safra tree for Streett determinization; $\delta_{RT}$ is a transition relation that is established during the construction of H-Safra trees for Streett determinization, consisting of transitions (typically $\delta$) which are of the form $H\xrightarrow[\delta_{\sig}]{\sigma}\hat{H}$ where $\delta_{\sig}=(\sig_{acc},\sig_{rej})$ is the signature of the transition $\delta$, with $\sigma$ ranging over $\Sigma$, and $H$ ranging over $Q_{RT}$; and $\lambda_{RT}=\{(A_{I1},R_{I1}),\ldots,(A_{Ik},R_{Ik})\}$ is the Rabin acceptance condition. Note that, in each Rabin pair $(A_{I},R_{I})$, $I$ ranges over the names appearing in the H-Safra trees for Streett determinization. $A_{I}$ is the set of transitions through which node $\tau$ with name being $I$ is accepting, while $R_{I}$ is the set of transitions through which node $\tau$ with name being $I$ is rejecting.

Given an input $\omega$-word $\alpha$: $\omega\rightarrow\Sigma$, we call the sequence $\Pi=(H_0,\alpha(0),H_1)$ $(H_1,\alpha(1),H_2)(H_2,\alpha(2),H_3)\ldots$ of transitions where $H_0=H_I$, and for all $i\in\omega$, $H_{i+1}$ is the $\alpha(i)$-successor of $H_i$, the \emph{H-Safra Streett trace} of the NSA $S$ over $\alpha$. We view the H-Safra Streett trace of $S$ over $\alpha$ as the \emph{run} of the DRTA $RT$ over $\alpha$. Then we say that $\alpha$ is \emph{accepted by the DRTA} if $\mathsf{Inf}(\Pi)\cap A_{Ii}\neq\emptyset$ and $\mathsf{Inf}(\Pi)\cap R_{Ii}=\emptyset$ for some $(A_{Ii},R_{Ii})$.

Let $RT$ be the DRTA obtained from the given NSA $S$. Theorem \ref{thm:RT} is formalized and proved.

\begin{Thm}\label{thm:RT}
$L(RT)=L(S)$.
\end{Thm}
\begin{proof}
This proof is similar to the one in \cite{CZ12}.

$\Leftarrow$: {This part of proof is almost identical to the one in \cite{Safra92}. We ought to show that if $\Pi=H_0H_1\cdots$ is a run of $RT$ over an infinite word $\alpha=\alpha_0\alpha_1\cdots\in L(S)$, then (1) a node $\tau$ exists in every state in $\Pi$ from some point on, (2) $\tau$ turns accepting infinitely often, and (3) $\tau$ has a fixed name $Ii$. The argument in \cite{Safra92} guarantees the existence of such a node $\tau$ with the first two properties. The only complication comes from renaming. We have the situation that $\tau$ with name $Ii$ exists in $H_m$, but it is renamed to $Ii'$ in the succeeding state $H_{m+1}$. This happens when the left sibling $\tau'$, whose index label $h(\tau')=h(\tau)$, of $\tau$ in $H_m$ is removed from $H_{m+1}$. However, it can only happen to $\tau$ finitely many times, as the left siblings with the same index labels of $\tau$ are finite and the new created siblings whose index labels are the same as $\tau$ will be placed to the right of $\tau$. Therefore, $\tau$ is eventually assigned a fixed name $Ii$, which provide us the third property.
}

$\Rightarrow$: {
Given an $\omega$-word $\alpha=\alpha_0\alpha_1\cdots\in L(RT)$, there exists an accepting run $\Pi=H_0H_1\cdots$ of $RT$ over $\alpha$. We ought to show that there is also an accepting run of $S$ over $\alpha$. $\Pi$ is accepting means that there exists an $Ii\in I$ such that $\Pi$ eventually never visits $R_{Ii}$, but visits $A_{Ii}$ infinitely often. Since renamed nodes or deleted nodes are rejecting, all nodes named by $Ii$ have to be the same node. It follows that a node $\tau$ eventually stays in every state in a suffix of $\Pi$ and $\tau$ turns accepting infinitely often. The rest of the proof is the same as the one in \cite{Safra92}.}
\end{proof}

\begin{Thm}\label{thm:complexity}
Given a NSA $S$ with $n$ states and $k$ Streett pairs, we can construct a DRTA with $n^{5n}(n!)^n$ states, $O(n^{n^2})$ Rabin pairs for $k=\omega(n)$ and $n^{5n}k^{nk}$ states, $O(k^{nk})$ Rabin pairs for $k=O(n)$ that recognizes the language $L(S)$.
\end{Thm}

\begin{proof}
For the state complexity,
by Lemma \ref{complex:HST}, we can calculate the number of structural ordered trees with state and index labels (i.e. $\mu$-Safra trees for Streett determinization without names, $E$ and $F$).
 {According to the result in  \cite{CZ12,CZ13}, there are at most $n^{4n}$ structural ordered trees. For every structural ordered tree, there are at most $n^n$ possibilities of state labeling. Besides, the number of possibilities of index labeling is bounded by $(n!)^n$ for $k=\omega(n)$, and $k^{nk}$ for $k=O(n)$. Thus, the state complexity is $n^{4n}\cdot n^n\cdot (n!)^n = n^{5n}(n!)^n$ for $k=\omega(n)$, and $n^{4n}\cdot n^n\cdot k^{nk} = n^{5n}k^{nk}$ for $k=O(n)$.}


For the index complexity, we have that for any branch from the root to a leaf of a H-Safra tree, there are at most $\mu$ nodes, say $\tau$, such that $j(\tau)\neq 0$. Moreover, a H-Safra tree contains at most $n$ nodes, say $\tau$, with $j(\tau)=0$ \cite{CZ12}.
Therefore, there are at most $n+\mu$ nodes in a branch. The name of a node is denoted by $x_1^{y_1}.x_2^{y_2}.\cdots .x_{n+\mu}^{y_{n+\mu}}$, where $x_i\in\{0,j_1,j_2,\ldots,j_{\mu}\}$ ($j_m$ is obtained by $Mini$ for $1\leq m\leq \mu$) and $y_i\in\{1,2,\ldots,n\}$. The number of $i$ such that $x_i=0$ is exactly $n$.  Thus, the number of names is
$$\dbinom{n}{n+\mu}\cdot n^n\cdot (\mu !)^n = O(\mu^{n\mu}).$$
{Since $\mu=\min(n,k)$, for $k=\omega(n)$, by replacing $\mu$ with $n$, the index complexity $O(n^{n^2})$ is obtained; for $k=O(n)$, by replacing $\mu$ with $k$, $O(k^{nk})$ is obtained.}
\end{proof}

\subsection{Construction of LIR-H-Safra Trees for Streett Determinization}\label{sec:NSA2DPTA}
Fix a NSA $S=(\Sigma,Q,Q_0,\delta,\langle G,B\rangle_k)$. The \emph{initial LIR-H-Safra tree for Streett determinization} $LH_I$ of $S$ is $H_I$ with a LIR. The order of all nodes in the LIR follows the order a node is generated.

Given a LIR-H-Safra tree $LH$ of $S$ and a $\sigma\in\Sigma$, we construct a new LIR-H-Safra tree $\hat{LH}$, called the $\sigma$-successor of $LH$, and the signature $\sig$ of the transition, also in six steps similar to the transformation from NSA to DRTA.
{
The differences are: (1) For a node $\tau$ in $LH$, if $p(\tau)$ changes during the transformation, $\tau$ is rejecting; otherwise, $\tau$ is stable.
(2) The signature is defined by $\sig=(st,p)$. If there is no accepting or rejecting node, $\sig=\emptyset$. Otherwise, in the case $\hat{\tau}$ is the node with the minimal position in the LIR among accepting or rejecting nodes in the transformation,  it has $p=p(\hat{\tau})$, $st:=\acc$ if $\hat{\tau}$ is accepting, and $st:=\rej$ if $\hat{\tau}$ is rejecting.}
{As a result, an equivalent DPTA $PT=(\Sigma,Q_{PT},Q_{PT0},\delta_{PT},\lambda_{PT})$ can be obtained. Here $Q_{PT}$ is the set of LIR-H-Safra trees for Streett determinization w.r.t $S$; $Q_{PT0}$ is the initial LIR-H-Safra tree for Streett determinization; $\delta_{PT}$ is a transition relation that is established during the construction of LIR-H-Safra trees for Streett determinization, consisting of transitions (typically $\delta$) which are quintuples of the form $LH\xrightarrow[\delta_{\sig}]{\sigma}\hat{LH}$ where $\delta_{\sig}$ is the signature of the transition $\delta$, with $\sigma$ ranging over $\Sigma$, and $LH$ ranging over $Q_{PT}$; $\lambda_{RT}=\{\lambda_2,\lambda_3,\cdots,\lambda_{2n(\mu+1)},\lambda_{2n(\mu+1)+1}\}$ is the parity acceptance condition. Notice that for each $1\leq i \leq 2n(\mu+1)$,
$$
\begin{array}{llll}
{\lambda_{2i}}&:=& \set{\delta\in\delta_{PT} \mid \delta_{\sig}=(\acc,i)}\\

{\lambda_{2i-1}}&:=& \set{\delta\in\delta_{PT} \mid \delta_{\sig}=(\rej,i)}\\

{\lambda_{2n(\mu+1)+1}}&:=& \set{\delta\in\delta_{PT} \mid \delta_{\sig}=\emptyset \mbox{ or } \delta_{\sig}=(\rej,1)}
\end{array}
$$
}

Given an input $\omega$-word $\alpha : \omega\rightarrow\Sigma$, we call the sequence $\Pi=(LH_0,\alpha(0),LH_1)$ $(LH_1,\alpha(1),LH_2)(LH_2,\alpha(2),LH_3)\ldots$ of transitions such that $LH_0=LH_I$, and for all $i\in\omega$, $LH_{i+1}$ is the $\alpha(i)$-successor of $LH_i$, the \emph{LIR-H-Safra Streett trace} of the NSA $S$ over $\alpha$. We view the LIR-H-Safra Streett trace of $S$ over $\alpha$ as the \emph{run} of the DPTA $PT$ over $\alpha$. Then we say that $\alpha$ is \emph{accepted by the DPTA} if the minimal index $k$ for which $\mathsf{Inf}(\Pi)\cap \lambda_k \neq\emptyset$ is even.

Let $PT$ be the DPTA obtained from the given NSA $S$. Theorem \ref{thm:PT} is formalized.

\begin{Thm}\label{thm:PT}
$L(PT)=L(S)$.
\end{Thm}

\begin{proof}
As it has been proved that $S$ is equivalent to the DRTA $RT$  in Section~\ref{sec:NSA2DRTA}, we further prove this theorem by showing $L(PT)=L(RT)$.

$\Leftarrow$: {Given an $\omega$-word $\alpha\in L(RT)$, there is a node $\tau$ that is accepting infinitely often and its name keeps unchanged eventually in the H-Safra Streett trace about $\alpha$. {It indicates that the position of $\tau$ in the LIR is non-increasing. Note that the position of $\tau$ in the LIR decreases when a node $\hat{\tau}$ at a smaller position with $h(\hat{\tau})\neq h(\tau)$ is removed.
However, this can only happen for finitely many times.} The node $\tau$ will eventually remain in the same position $p$ in the LIR and every node $\tau'$ with $p(\tau')\leq p$ will be stable. Hence, no odd priority $<2p$ occurs infinitely often. And from that time onward, the node $\tau$ is accepting infinitely many times.
Therefore, the smallest priority occurring infinitely often is even. It indicates that $\alpha\in L(PT)$.
}

$\Rightarrow$: {Let $\alpha$ be an $\omega$-word in $L(PT)$. There is a LIR-H-Safra Streett trace $\Pi$ and an index $2i$ such that $\mathsf{Inf}(\Pi)\cap\lambda_{2i}\neq\emptyset$ and $\mathsf{Inf}(\Pi)\cap\lambda_k = \emptyset$ for any $k<2i$. It indicates that each node $\tau$ with $p(\tau)\leq i$ remains stable in the LIR from a time onward. That is $\tau$ is not rejecting. Meanwhile, the node on position $i$ is accepting infinitely often from that time onward. Thus $\alpha\in L(RT)$.
}
\end{proof}

\begin{Thm}\label{thm:complexity-NS2PT}
Given a NSA $S$ with $n$ states and $k$ Streett pairs, we can construct a DPTA with $3(n(n+1)-1)!(n!)^{n+1}=2^{O(n^2 \log n)}$ states, $2n(n+1)$ priorities for $k=\omega(n)$ and $3(n(k+1)-1)!n!k^{nk}=2^{O(nk \log nk)}$ states, $2n(k+1)$ priorities for $k=O(n)$ that recognizes the language $L(S)$.
\end{Thm}

\begin{proof}
The number of nodes in a LIR-H-Safra tree is also at most $n(\mu+1)$. Similar to the analysis in \cite{Sven09}, there are at most $(n(\mu+1)-1)!$ LIR-H-Safra trees without state and index labels.
For the state labelling function, {let $t(n,m)$ denote the number of LIR-H-Safra trees without index labels, say $\tilde{LH}$, such that there are $m$ nodes in $\tilde{LH}$ and $n$ states in the state label of the root of $\tilde{LH}$}. Fist, we have $t(n,n(\mu+1))=(n(\mu+1)-1)!n!$. A conclusion has been proved in \cite{Sven09} that for every $m\leq n(\mu+1)$, $t(n,m-1)\leq\frac{1}{2}t(n,m)$.
Hence, $\sum_{i=1}^{n(\mu+1)}t(n,i)\leq 2(n(\mu+1)-1)!n!$.
If there are $n'$ $(n'<n)$ states labelled in the root, the number of the LIR-H-Safra trees without index labels is $2(n'(\mu'+1)-1)!n'!\tbinom{n'}{n}\leq 2(n'(\mu'+1)-1)!n!$, where $\mu'=\min(n',k)$.
Thus, the number of LIR-H-Safra trees without index label is $\sum_{n'=1}^{n}2(n'(\mu'+1)-1)!n! \leq 3(n(\mu+1)-1)!n!$. 
By the result in  \cite{CZ12,CZ13}, the number of possibilities of index labeling is bounded by $(n!)^n$ for $k=\omega(n)$, and $k^{nk}$ for $k=O(n)$. It follows that the number of LIR-H-Safra trees is at most $3(n(n+1)-1)!(n!)^{n+1} = 2^{O(n^2 \log n)}$ for $k=\omega(n)$ by replacing $\mu$ with $n$ and $3(n(k+1)-1)!n!k^{nk} = 2^{O(nk \log nk)}$ for $k=O(n)$ by replacing $\mu$ with $k$. 
\end{proof}

\section{Lower Bound Complexity}\label{sec:lower bound}
As for the state lower bound, it means the minimum states required by the equivalent deterministic automata, regardless of whether the acceptance condition is state-based or transition-based.
In this section, we prove a lower bound state complexity for determinization construction from NSA to DR(T)A, which exactly matches the state complexity of the proposed  determinization construction.
Further, we put forward a lower bound state complexity for determinization construction from NSA to DP(T)A, which is the same as the state complexity of the proposed determinization construction in the exponent.

\subsection{$L$-Game}
\begin{Def}[$L$-game \cite{CZ09}]\rm  
An \emph{$L$-game} for two players, Adam and Eva, is a tuple $\mathcal{G}=(V,V_E,V_A,p_I,\Sigma,\text{Move},L)$, where
\begin{itemize}
  \item $V$ is a set of \emph{positions} which is partitioned into the \emph{positions for Eva} $V_E$ and the \emph{positions for Adam} $V_A$,
  \item $p_I\in V$ is the \emph{initial position} of $\mathcal{G}$,
  \item $\Sigma$ is the \emph{labelling alphabet},
  \item $\text{Move}\subseteq V\times\Sigma\times V$ is the set of possible \emph{moves}, and
  \item $L\subseteq\Sigma^{\omega}$ is the \emph{winning condition}.
\end{itemize}
\end{Def}

A tuple $(p,\sigma,p')\in\text{Move}$ indicates that there is a move from $p$ to $p'$, which produces a letter $\sigma$. A \emph{play} is a maximal sequence $\pi=(p_0,\sigma_0,p_1,\sigma_1,p_2,\sigma_2,\ldots)$ such that $p_0=p_I$, and for each $i$, $(p_i,\sigma_i,p_{i+1})\in\text{Move}$. The player who belongs to the current position will choose the next move. Let $\pi_{\Sigma}=(\sigma_0,\sigma_1,\sigma_2,\ldots)$. If $\pi_{\Sigma}\in L$, Eva wins the play. Otherwise, Adam wins the play.

A \emph{strategy for the player X} is a function which tells the player what move he should choose depending on the finite history of moves played so far. A strategy is called a \emph{winning} strategy for Eva (resp. Adam), if Eva (resp. Adam) wins every play with this strategy. A \emph{strategy with memory m for Eva} is described as ($M$, update, choice, init), {in which $M$ is a set of \emph{memory} with the size being $m$}, update is a mapping from $M\times\text{Move}$ to $M$, choice is a mapping from $V_E\times M$ to Move, and $\text{init}\in M$. A player $X$ wins a game with memory $m$ if it has a winning strategy with memory $m$.

The following Lemma proved in \cite{CZ09} provides an argument for proving lower bounds on determinization problems.

\begin{Lem} \rm\label{proof-foundation}
If Eva wins an $L$-game, and requires memory $m$ for that, then every deterministic Rabin automaton for $L$ has states at least $m$ \cite{CZ09}.
\end{Lem}

\subsection{Lower Bound State Complexity for NSA to DR(T)A} \label{LBNS2DR}
Inspired by the approach in \cite{CZ09}, in order to prove the lower bound state complexity for the determinization construction from NSA to DR(T)A, the essence is to define \emph{full Streett automata} and the relevant game.

For convenience, we first introduce some notations. For a tree $T$, every node $\tau\in T$ can be expressed by a sequence $se(\tau)=se(\tau)(0)se(\tau)(1)se(\tau)(2)\cdots$, where $se(\tau)(i)$ $(i\geq 0)$ is a positive integer. For the root $\tau_r$ of $T$, we have $se(\tau_r)=1$. As for any other node $\tau$, $se(\tau)=se(\tau_p)i$, where $\tau_p$ is the parent of $\tau$ and $i=1+|\{\tau'\in T\mid\tau'$ is the left sibling of $\tau\}|$. {For any two nodes $\tau$ and $\tau'$, we define $\tau<_{lex}\tau'$ if $se(\tau)$ is the proper prefix of $se(\tau')$; or there exists $i$ such that $se(\tau)(i)<se(\tau')(i)$ and for all $j<i$, $se(\tau)(j)=se(\tau')(j)$. Further, $\tau\leq_{lex}\tau'$ if $\tau<_{lex}\tau'$ or $se(\tau)=se(\tau')$.}

\begin{Def} [Full Streett Automata] \rm
A \emph{full Streett automaton} is a quintuple $(Q,\Sigma,Q_0,$ $\delta,\langle G,B\rangle_{[k]})$ where $Q$ is a finite set of states, $Q_0\subseteq Q$ is a set of initial states, $\Sigma=\mathcal{P}(Q\times\{\emptyset,G_1,\ldots,G_k,B_1,\ldots,B_k\}\times Q)$ is the alphabet, and the transition relation is defined by $\delta\subseteq Q\times\Sigma\times Q$. $\langle G,B \rangle_{[k]}$ are Streett pairs, where $k$ is a positive integer, and $G_i$ and $B_i$ are sets of transitions for $1\leq i\leq k$. For a Streett pair $\langle G_i,B_i \rangle$ and a letter $\sigma\in\Sigma$, a transition $\delta = (p,\sigma,q)\in G_i$ (or $B_i$) iff $(p,G_i \mbox{ (or }B_i),q)\in\sigma$.
\end{Def}

For the full Streett automaton with $n$ states $\mathcal{S}_n=(Q,\Sigma,Q,\delta,$ $\langle G,B\rangle_{[k]})$, where $Q$ is also the set of initial states, and $L(\mathcal{S}_n)=L_n$. A DRTA $\mathcal{RT}=(Q_{RT},\Sigma,Q_{RT0},\delta_{RT},\lambda_{RT})$ can be constructed via \emph{H-Safra trees for Streett determinization}.

We introduce some useful notations. For a set of states $S\subseteq Q$, let $\Sigma_{S}$ be the set of letters $\sigma\in\Sigma$ such that $\bigcup_{q \in S} \delta(q,\sigma) = S$. We also let $L_n^S=L_n\cap\Sigma_S^{\omega}$ and $Q^S_{RT}=\{H\in Q_{RT}:l(\tau_r)=S\mbox{ where }\tau_r\mbox{ is the root of }H\}$.
Thus, for all words $u\in\Sigma_S^*$ and all $H\in Q^S_{RT}$, we have $\delta_{RT}(H,u)\in Q^S_{RT}$.

Given a set of states $S\subseteq Q$, we define a $L_n^S$-$S$-game $\mathcal{G}^S$ such that Eva wins $\mathcal{G}^S$ but she cannot win with memory less than $|Q^S_{RT}|$. This indicates that any determinization Rabin automaton accepting $L_n^S$ has at least $|Q^S_{RT}|$ states.

\begin{Def}[$L_n^S$-$S$-game] \rm \label{Def:L-S-game}
The $L_n^S$-$S$-game is a tuple $\mathcal{G}^S=(V,V_E,V_A,p_I,\Sigma_S^+,\mbox{Move},L_n^S)$, where $V_E$ is a singleton set $\{p_E\}$ and $V_A$ consists of the initial position $p_I$ and one position $p_H$ for each H-Safra tree $H\in Q^S_{RT}$. The Move of $\mathcal{G}^S$ includes:

\begin{itemize}
  \item $(p_I,u,p_E)$, $u$ is a non-$\epsilon$ word in $\Sigma_S^+$.
  \item $(p_E,\epsilon,p_H)$, for each H-Safra tree $H$ in $Q^S_{RT}$.
  \item $(p_H,u,p_E)$, if there exists a node $\hat{\tau}$ in $\hat{H}=\delta_{RT}(H,u)$ that satisfies one of the three following conditions during the transformation from $H$ to $\hat{H}$:
      \begin{enumerate}
        \item $\hat{\tau}$ is accepting, and for all $\hat{\tau}'\leq_{lex}\hat{\tau}$ in $\hat{H}$, $\hat{\tau}'$ is not rejecting, $h(\hat{\tau}')=h(\tau')$ and $l(\hat{\tau}')=l(\tau')$, 

        \item $j(\hat{\tau})<j(\tau)$, and for all $\hat{\tau}'<_{lex}\hat{\tau}$ in $\hat{H}$, $\hat{\tau}'$ is not rejecting, $h(\hat{\tau}')=h(\tau')$, and $l(\hat{\tau}')=l(\tau')$, 

        \item $j(\hat{\tau})=j(\tau)$, $l(\hat{\tau})\supset l(\tau)$, and for all $\hat{\tau}'<_{lex}\hat{\tau}$ in $\hat{H}$, $\hat{\tau}'$ is not rejecting, $h(\hat{\tau}')=h(\tau')$ and $l(\hat{\tau}')=l(\tau')$, 
      \end{enumerate}
      for each H-Safra tree $H$ in $Q^S_{RT}$ and a word $u\in\Sigma_S^+$.
      Note that $\tau$ and $\tau'$ are the nodes in $H$ with $se(\tau)=se(\hat{\tau})$ and $se(\tau')=se(\hat{\tau}')$, respectively.
\end{itemize}
\end{Def}

The $L_n^S$-$S$-game has a flower shape, which is intuitively illustrated in Fig.~\ref{fig:L-S-game}. The central position is controlled by Eva and the petals belong to Adam. Moreover, each petal corresponds to a \emph{H-Safra tree}.

\begin{figure}[htp]
\centerline{\includegraphics[height=4.8cm]{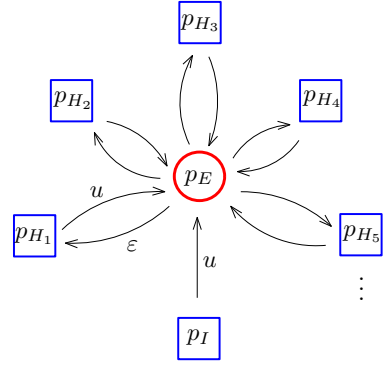}}\caption{The $L_n^S$-$S$-game $\mathcal{G}^S$} \label{fig:L-S-game}
\end{figure}

\begin{Lem} \rm \label{Eva-win-S}
Eva has a winning strategy in $\mathcal{G}^S$.
\end{Lem}

\begin{proof}
There is a winning strategy for Eva: if a word $u$ was produced after a finite play and Eva is to make a move from $p_E$, then she chooses to go to a position indexed by $\delta_{RT}(H_0,u)$ where $H_0=Q_{RT0}$.

To see that Eva wins the $L_n^S$-$S$-game $\mathcal{G}^S$ with this strategy, we consider the run $\rho_{RT}$ of $\mathcal{RT}$ on the word defined by the play $(p_I,u_0,p_E)(p_E,\varepsilon,p_{H_1})(p_{H_1},u_1,p_E)(p_E,\varepsilon,p_{H_2})\cdots$, which refers to the word $u_0u_1u_2\cdots$.
Each segment $\rho_{RT}(H_i,u_i,H_{i+1})$ ($i\geq 1$) of the run $\rho_{RT}$ and a corresponding node $\tau_i\in H_{i+1}$ satisfie one of the conditions 1, 2 and 3 in Definition \ref{Def:L-S-game}. We denote $\tau=\leq_{lex}$-$\min\{\tau_i\mid\tau_i\mbox{ occurs infinitely often}\}$, then each $\tau'$, such that $\tau'\leq_{lex}\tau$, is not rejecting in each segment of $\rho_{RT}$. Obviously, if $\tau$ is infinitely often accepting, then Eva wins.

Assume that there is a position in $\rho_{RT}$ such that $\tau$ is not accepting, but the value of $j(\tau)$ becomes smaller infinitely often from the position onwards. However, this can only happen finitely often since $j(\tau)$ has the minimal value $0$, which is a contradiction.

 Also, assume that from some position in $\rho_{RT}$ onwards, $\tau$ is not accepting and the index label remains constant. Nevertheless, the state label $l(\tau)$ would grow monotonously and would infinitely often grow strictly. It can only happen finitely many times since $l(\tau)\subseteq l(\tau_p)$, which is a contradiction.

Therefore, Eva wins $\mathcal{G}^S$ with this strategy.
\end{proof}

Next, for each H-Safra tree $H\in Q^S_{RT}$, a game $\mathcal{G}_H^S$ is defined, which is a modification of $\mathcal{G}^S$ by removing the position $p_H$ of Adam and the corresponding moves. For this game, the following Lemma holds.


\begin{Lem} \rm \label{Adam-win-foundation-S}
For any two H-Safra trees $H\neq H'$ in $Q^S_{RT}$, there exists a word $u$ such that
$(p_{H'},u,p_E)$ is a move in $\mathcal{G}_H^S$, $\delta_{RT}(H',u)=\delta_{RT}(H,u)=H$, and
for any node $\tau$ in $H$, $\tau$ is not accepting.
\end{Lem}

\begin{proof}
This lemma requires an analysis of the differences between the two H-Safra trees $H$ and $H'$. 
For the $\leq_{lex}$-minimal nodes $\tau$ in $H$ and $\tau'$ in $H'$ where $se(\tau)=se(\tau')$, but $l(\tau)\neq l(\tau')$ or $h(\tau)\neq h(\tau')$, a letter $\sigma$ is defined first which has the following two cases, denoted as $\sigma'$ and $\sigma''$, respectively.

(i) If $\tau$ and $\tau'$ are the left most child of their parents $\tau_p$ and $\tau'_p$, respectively, $\sigma'$ is produced such that $\{(s,\emptyset,s_p)\mid s\in l(\tau)\cup l(\tau')\mbox{ and } s_p\in l(\tau_p)\}\subseteq\sigma'$.

(ii) If $\tau$ and $\tau'$ have left siblings $\tau_l$ and $\tau'_l$, respectively, it is apparent that $l(\tau_l)=l(\tau'_l)$ and $h(\tau_l)=h(\tau'_l)$. Then we construct $\sigma''$ such that $\{(s,\emptyset,s_p)\mid s\in l(\tau)\cup l(\tau')\mbox{ and } s_p\in l(\tau_p)\backslash l(\tau_l)\}\subseteq\sigma''$.

For these two cases, after reading $\sigma$ at $H$ and $H'$, we have $l(\tau)=l(\tau')$. Every node $\hat{\tau}<_{lex}\tau$ in $H$ and $\hat{\tau'}<_{lex}\tau'$ in $H'$ remain unchanged. Meanwhile, for each node $\tau_m >_{lex}\tau$ in $H$ and $\tau'_m >_{lex} \tau'$ in $H'$, we have $l(\tau_m)=\emptyset$ and $l(\tau'_m)=\emptyset$.



Next, for two different nodes $\tau$ and $\tau'$, there are four cases to be considered:


(1) $j(\tau)>j(\tau')$. In the case that $\tau$ and $\tau'$ are the left most child of their parents $\tau_p$ and $\tau'_p$, respectively, let $w=\sigma_{j(\tau')}\sigma_{j(\tau')-1}\cdots\sigma_1$. Here, for each $1\leq k \leq j(\tau')$, $\sigma_k = \{(s,G_k,s)\mid s\in l(\tau')\}$. By reading $\sigma' w$, $H$ and $H'$ can reach $\hat{H}$ and $\hat{H'}$, respectively. The parent of $\tau'$ is accepting and $\tau$ stays unchanged.
In the case that $\tau$ and $\tau'$ have left siblings $\tau_l$ and $\tau'_l$, respectively, it has $l(\tau_l)=l(\tau'_l)$ and $h(\tau_l)=h(\tau'_l)$.
Let $s$ be a state in $l(\tau_l)$. We construct a word $w=\sigma_{j(\tau_l)}\sigma_{j(\tau_l)-1}\cdots\sigma_{j(\tau)+1}$. Here, for each $j(\tau)+1\leq k\leq j(\tau_l)$, $(s,G_k,s)\in\sigma_k$.
By reading $\sigma'' w$, a new node $\tau_s$ is created as the sibling of $\tau_l$ with $l(\tau_s)=\{s\}$, $h(\tau_s)=h(\tau)$, and $\tau'_s$ is created as the sibling of $\tau'_l$ with $l(\tau'_s)=\{s\}$, $h(\tau'_s)=h(\tau)$. Then $\tau_s$ and $\tau'_s$ are accepting in the next transformation.
Later, let $l(\tau)=l(\tau')=\emptyset$ and remove $\tau$ and $\tau'$, which makes $\tau_s$ renamed (rejected), and $\tau'_s$ not rejected. After the above operations, $\hat{H}$ and $\hat{H'}$ are obtained, respectively.

(2) $j(\tau)<j(\tau')$. Construct a word $w = \sigma_{j(\tau')} \sigma_{j(\tau')-1} \cdots$ $\sigma_{j(\tau)+1}$. Here, for each $j(\tau)+1 \leq k \leq j(\tau')$, it has $\sigma_k = \{(s,G_k,s)\mid s\in l(\tau')\}$. By reading $\sigma w$, $H$ and $H'$ can reach $\hat{H}$ and $\hat{H'}$, respectively.

(3) $j(\tau)=j(\tau')$ and $l(\tau)\supset l(\tau')$. After reading $\sigma$ at $H$ and $H'$, $\hat{H}$ and $\hat{H'}$ are obtained, respectively.

(4) $j(\tau)=j(\tau')$ and $l(\tau') \setminus l(\tau) \neq \emptyset$. We first construct a word $w$, which makes $\tau'$ being accepting after reading $w$ at $H'$.
Then construct a letter $\hat{\sigma}$ such that $(s,B_{j(\tau)},s)\in\hat{\sigma}$ for each state $s\in l(\tau)$. As a consequence, by reading $w\hat{\sigma}$ at $H$ and $H'$, $\tau$ becomes rejected and $\tau'$ is accepting. Furthermore, $\hat{H}$ and $\hat{H'}$ are obtained, respectively.

For the four cases, the {next transformation} makes both $\hat{H}$ and $\hat{H'}$ move to $H$.

Therefore, in the transformation from $H'$ to $H$, (1) and (4) satisfy condition 1) of Definition \ref{Def:L-S-game}. What is more, (2) and (3) satisfy condition 2) and 3), respectively. Meanwhile, there exists no accepting node during the transformation from $H$ to $H$.
\end{proof}

Further, by Lemma \ref{Adam-win-foundation-S},  the following lemma is obtained.

\begin{Lem} \rm \label{Adam-win-S}
For every H-Safra tree $H$ in $Q^S_{RT}$, Adam has a winning strategy in the correspongding $\mathcal{G}^{S}_{H}$.
\end{Lem}

\begin{proof}
There is a winning strategy for Adam as follows. When he plays a word $u$ from $p_I$ such that $\delta_{RT}(H_0,u)$ where $H_0 = Q_{RT0}$, the best choice for Eva is to move to $p_H$ on the basis of the proof of Lemma \ref{Eva-win-S}. However, this position has been removed, she is forced to move to another position $p_{H'}$ ($H' \neq H$). Then Adam moves according to Lemma \ref{Adam-win-foundation-S}, and he can always answer to the proposal of Eva similarly in the play. Meanwhile, an infinite word $\alpha$ is produced. It is obvious that $\mathcal{RT}$ does not accept $\alpha$ because of Lemma \ref{Adam-win-foundation-S}. Therefore, Adam has a winning strategy in $\mathcal{G}^{S}_{H}$.
\end{proof}

Then, it is easy to infer the following lemma.

\begin{Lem} \rm \label{Eva-nowin-S}
Eva has no winning strategy with memory less than $|Q^S_{RT}|$ in $\mathcal{G}^{S}$.
\end{Lem}

\begin{proof}
For a contradiction, we suppose that Eva has a winning strategy with memory $|Q^S_{RT}|-1$. Then there would be a position $p_H$ which is never visited by this strategy. It is a contradiction with Lemma \ref{Adam-win-S}.
\end{proof}

Similar to the approach in \cite{CZ09}, the main theorem is ready to be proved.
\begin{Thm}
Every DR(T)A accepting $L(\mathcal{S}_n)$ has states at least $|Q_{RT}| = n^{5n}(n!)^n$ for $k=\omega(n)$ and $n^{5n}k^{nk}$ for $k=O(n)$.
\end{Thm}

This theorem means that the proved lower bound state complexity for the determinization construction from NSA to DR(T)A exactly matches the state complexity of the proposed determinization construction by H-Safra trees.

\subsection{Lower Bound State Complexity for NSA to DP(T)A}
To prove the lower bound state complexity for determinization construction from NSA to DP(T)A, an appropriate $L$-game, for recognizing the complement language of the NSA, is constructed first.


For the full Streett automaton $\mathcal{S}_n = (Q,\Sigma,Q,\delta,\langle G,B \rangle_{[k]})$, a DPTA $\mathcal{PT} = (Q_{PT},\Sigma,Q_{PT0},\delta_{PT},\lambda_{PT})$ can be constructed via \emph{LIT-H-Safra trees}.
Let $L_n ^c$ be the complement of $L(\mathcal{S}_n)$, $\Sigma_S^\omega$ denote the infinite words over $\Sigma_S$, and $L_n^{cS}=L_n^c\cap\Sigma_S^\omega$. For any $S\subseteq Q$, let $Q^S_{PT} = \{LH\in Q_{PT}:l(\epsilon)=S \mbox{ where } \epsilon \mbox{ is the root of}$ $LH\}$ be the set of LIR-H-Safra trees in which state label of the root is $S$.
We choose a subset $Q^{Sh}_{PT}$ of $Q^S_{PT}$, which satisfies:
For any two LIR-H-Safra trees $LH,LH'\in Q^{Sh}_{PT}$ and any nodes $\tau$ in $LH$, $\tau'$ in $LH'$, if $se(\tau)=se(\tau')$, then $h(\tau)=h(\tau')$.


Given a set of states $S\subseteq Q$, we define a $L_n^{cS}$-$S$-game $\mathcal{G}^{cS}$ such that Eva wins $\mathcal{G}^{cS}$ but she cannot win with memory less than $|Q^{Sh}_{PT}|$.

\begin{Def}[$L_n^{cS}$-$S$-game] \rm \label{Def:Lc-S-game}
The $L_n^{cS}$-$S$-game is a tuple $\mathcal{G}^{cS}=(V,V_E,V_A,p_I,\Sigma_S^+,\mbox{Move},L_n^{cS})$, where $V_E$ is a singleton set $\{p_E\}$ and $V_A$ consists of the initial position $p_I$ and one position $p_{LH}$ for each LIR-H-Safra tree $LH\in Q^{Sh}_{PT}$. The Move of $\mathcal{G}^{cS}$ includes:
\begin{itemize}
  \item $(p_I,u,p_E)$, $u$ is a non-$\epsilon$ word in $\Sigma_S^+$.
  \item $(p_E,\epsilon,p_{LH})$, for each LIR-H-Safra tree $LH$ in $Q^{Sh}_{PT}$.
  \item $(p_{LH},u,p_E)$, if there exists a node $\tau$ in $LH$ with $p(\tau)=i$ and $\tau$ satisfies one of the two following conditions in the transition from $LH$ to $\hat{LH}=\delta_{PT}(LH,u)$:
      \begin{enumerate}
        \item $\tau$ is rejecting and the priority of the transition is $2i-1$, and for each $\tau'$ in $LH$ such that $p(\tau')<p(\tau)$, it requires that $l(\hat{\tau}')=l(\tau')$ and $h(\hat{\tau}')=h(\tau')$,

        \item $h(\hat{\tau})=h(\tau)$, $l(\hat{\tau})\subset l(\tau)$, and the priority of the transition is larger than $2i$, and for each $\tau'$ in $LH$ such that $p(\tau')<p(\tau)$, it requires that $l(\hat{\tau}')=l(\tau')$ and $h(\hat{\tau}')=h(\tau')$,
      \end{enumerate}
      for each LIR-H-Safra tree $LH$ in $Q^{Sh}_{PT}$ and a word $u\in\Sigma_S^+$.
      Note that $\hat{\tau}$ and $\hat{\tau}'$ are nodes in $\hat{LH}$ with $p(\hat{\tau})=p(\tau)$ and $p(\hat{\tau}')=p(\tau')$, respectively.
\end{itemize}
\end{Def}

\begin{Lem} \rm \label{Eva-win-S2P}
Eva has a winning strategy in $\mathcal{G}^{cS}$.
\end{Lem}
\begin{proof}
There is a winning strategy for Eva: if a word $u$ was produced after a finite play and Eva is to make a move from $p_E$, then she chooses to go to a position indexed by $\delta_{PT}(LH_0,u)$ where $LH_0=Q_{PT0}$.

To see that Eva wins the $L_n^{cS}$-$S$-game $\mathcal{G}^{cS}$ with this strategy, we consider the run $\rho_{PT}$ of $\mathcal{PT}$ on the word defined by the play $(p_I,u_0,p_E)(p_E,\epsilon,p_{LH_1})(p_{LH_1},u_1,p_E)(p_E,\epsilon,p_{LH_2})\cdots$,\\which refers to the word $u_0u_1u_2\cdots$. Each segment $\rho_{PT}(LH_k,$ $u_k,LH_{k+1})$ $(k\geq1)$ of the run $\rho_{PT}$ satisfies one of the conditions 1), 2) and 3) in Definition \ref{Def:Lc-S-game}, and there exists a node $\tau_k\in LH_k$ with $p(\tau_k)=i_k$. Let $i_{\min}$ be the minimal one that occurs infinitely often among these $i_k$ and $\tau_{\min}$ be the node on position $i_{\min}$ in the LIR. Hence, no priority smaller than $2i_{\min}-1$ can occur infinitely often in $\rho_{PT}$.
It is obvious that if $\tau_{\min}$ is infinitely often rejecting, then the minimal priority occurring infinitely often is $2i_{\min}-1$ in $\rho_{PT}$, and Eva wins.

Next, we assume that there is a position $po$ of $\rho_{PT}$ such that $\tau_{\min}$ is not rejecting, but satisfies condition 2) in Definition \ref{Def:Lc-S-game} infinitely often from the position $po$ onwards. Consequently, the state label of $\tau_{\min}$ would reduce monotonously from the position $po$ onwards, and would infinitely often reduce strictly. It is a contradiction.



Therefore, Eva wins $\mathcal{G}^{cS}$ with this strategy.
\end{proof}

Similar to the lower bound state complexity from NSA to DR(T)A in Section \ref{LBNS2DR}, for each LIR-H-Safra tree $LH\in Q^{Sh}_{PT}$, a game $\mathcal{G}_{LH}^{cS}$ can be defined by removing the corresponding position $p_{LH}$ and the relevant moves from $\mathcal{G}^{cS}$.
The following lemma shows that Adam has a winning strategy in $\mathcal{G}_{LH}^{cS}$.

\begin{Lem} \rm \label{Adam-win-foundation-S2P}
For any two LIR-H-Safra trees $LH\neq {LH}'$ in $Q^{Sh}_{PT}$, there exists a word $u$ such that $(p_{LH'},u,p_E)$ is a move in $\mathcal{G}_{LH}^{cS}$, $\delta_{PT}(LH',u)=\delta_{PT}(LH,u)=LH$, and the minimal priority in the transitions from $LH$ to $LH$ after reading $u$ is even.
\end{Lem}

\begin{proof}
We first identify the position-minimal nodes $\tau$ in $LH$ and $\tau'$ in $LH'$ such that $p(\tau)=p(\tau')=i$, and $se(\tau)\neq se(\tau')$ or $l(\tau)\neq l(\tau')$.
We use $W$ to denote a set of words such that for each $w \in W$, $\tau$ is accepting and after reading $w$ at $LH$, the priority is $2i$. Then two cases are considered:

(1) $se(\tau)=se(\tau')$. It has $h(\tau)=h(\tau')$. The only difference between $\tau$ and $\tau'$ is the state labels.
In the case that $l(\tau) \setminus l(\tau') \neq \emptyset$, a word $w_1 \in W$ is read at $LH$ and $LH'$. Let $\sigma$ be a letter such that $(s,B_{j(\tau')},s)\in\sigma$, where $s\in l(\tau')$. By reading $w_1 \sigma$, $LH$ and $LH'$ can reach $\hat{LH}$ and $\hat{LH'}$, respectively. In the transformation from $LH$ to $\hat{LH}$, $\tau$ is accepting and the priority is $2i$. Meanwhile, $\tau'$ is rejecting and the priority is $2i-1$ in the transformation from $LH'$ to $\hat{LH'}$. The next transformation makes both $\hat{LH}$ and $\hat{LH'}$ move to $LH$.
In the case that $l(\tau)\subset l(\tau')$, let $w_2$ be a word in $W$ such that $\tau'$ is not accepting or rejecting after reading $w_2$ at $LH'$. As a result, $\hat{LH}$ and $\hat{LH'}$ are obtained. The next transformation makes both $\hat{LH}$ and $\hat{LH'}$ move to $LH$ and the priority is larger than $2i$ in the transformation from $\hat{LH'}$ to $LH$.

(2) $se(\tau)\neq se(\tau')$. Let $w_3$ be a word in $W$ such that $\tau'$ is rejecting and after reading $w_3$ at $LH'$, the priority is $2i-1$. Then, $\hat{LH}$ and $\hat{LH'}$ are obtained, respectively. The next transformation makes both $LH$ and $LH'$ move to $LH$.

As a result, in the transformation from $LH'$ to $LH$, the first case of (1) and (2) satisfy condition 1) of Definition \ref{Def:Lc-S-game}. The second case of (1) satisfies condition 2). Meanwhile, the minimal priority is $2i$ in the transformation from $LH$ to $LH$.
\end{proof}





Thus, Eva has no winning strategy with memory less than $|Q^{Sh}_{PT}|$ in $\mathcal{G}^{cS}$. Based on the approach in \cite{CZ09} and Lemma \ref{proof-foundation},
we can obtain the following result.

\begin{Lem} \rm
Every DR(T)A that recognises the complement of $L(\mathcal{S}_n)$ must contain at least $|\bigcup\limits_{S\subseteq Q} Q^{Sh}_{PT}|$ states.
\end{Lem}

In \cite{Sven09}, there is a result that the size of the smallest Rabin automaton that recognises the complement of $L(\mathcal{S}_n)$ is equal to the one of the smallest Streett automaton that recognises $L(\mathcal{S}_n)$. Since parity automata are special Streett automata, the main theorem is inferred.




\begin{Thm}
Every DS(T)A or DP(T)A accepting $L(\mathcal{S}_n)$ must have states at least $|\bigcup\limits_{S\subseteq Q} Q^{Sh}_{PT}| = 2^{\Omega(n^2 \log n)}$ for $k=\omega(n)$ and $2^{\Omega(nk \log nk)}$ for $k=O(n)$.
\end{Thm}

Finally, we give the estimate for $|\bigcup\limits_{S\subseteq Q} Q^{Sh}_{PT}|$. Since the index label of each node is fixed, we can neglect the impact of the index label. Therefore, by the proof of Theorem \ref{thm:complexity-NS2PT}, we have
$$|\bigcup\limits_{S\subseteq Q} Q^{Sh}_{PT}| = 3(n(\mu+1)-1)!n!.$$
Specifically, $|\bigcup\limits_{S\subseteq Q} Q^{Sh}_{PT}| = 3(n(n+1)-1)!n! = 2^{\Omega(n^2 \log n)}$ for $k=\omega(n)$ by replacing $\mu$ with $n$, and $3(n(k+1)-1)!n! = 2^{\Omega(nk \log nk)}$ for $k=O(n)$ by replacing $\mu$ with $k$.

By the result in Section \ref{sec:NSA2DPTA}, the state complexity for the construction from NSA to DPTA is $3(n(n+1)-1)!(n!)^{n+1} = 2^{O(n^2 \log n)}$ for $k=\omega(n)$ and $3(n(k+1)-1)!n!k^{nk} = 2^{O(nk \log nk)}$ for $k=O(n)$. So, the above lower bound  is the same as the  upper bound in the exponent. There is still a slight gap between the lower and upper bounds.

\section{Conclusion} \label{Sec:con}
In this paper, we present determinization transformations from NSA with $n$ states and $k$ Streett pairs to DRTA with $n^{5n}(n!)^n$ states, $O(n^{n^2})$ Rabin pairs for $k=\omega(n)$ and $n^{5n}k^{nk}$ states, $O(k^{nk})$ Rabin pairs for $k=O(n)$; and to DPTA with $3(n(n+1)-1)!(n!)^{n+1}$ states, $2n(n+1)$ priorities for $k=\omega(n)$ and $3(n(k+1)-1)!n!k^{nk}$ states, $2n(k+1)$ priorities for $k=O(n)$.
Further, we prove a lower bound state complexity for determinization construction from NSA to DR(T)A, which matches the state complexity of the proposed  determinization construction.
Also, we put forward a lower bound state complexity for determinization construction from NSA to DP(T)A  which is the same as the proposed determinization construction in the exponent.

In the near future, we will implement the proposed determinization constructions and evaluate efficiency of the algorithms in practice.


\begin{thebibliography}{1}
\bibitem{Streett82}R.S.Streett. Propositional dynamic logic of looping and converse. \emph{Information and Control}, 54:121-141, 1982.

\bibitem{Buchi62}J. R. B\"{u}chi. On a decision method in restricted second order
arithmetic. In \emph{Proceedings of the International Congress on
Logic, Method, and Philosophy of Science}, pages 1-12. Stanford
University Press, 1962.

\bibitem{Safra92}Safra, S.: Exponential Determinization for omega-Automata with Strong-Fairness Acceptance Condition (Extended Abstract). STOC 1992: 275-282.


\bibitem{SV89} S. Safra, M. Y. Vardi, On $\omega$-automata and temporal logic, in: Proceedings of the 21st annual ACM symposium on Theory of computing (STOC'89), ACM,
1989, pp. 127-137.


\bibitem{CZ12} Y. Cai, T. Zhang: Can nondeterminism help complementation? In GandALF, pages 57-70, 2012.

\bibitem{CZ13} Y. Cai, T. Zhang: Determinization complexities of $\omega$ automata, Technical report (2013), http://theory.stanford.edu/~tingz/tcs.pdf

\bibitem{CZ11(1)} Y. Cai, T. Zhang: Tight Upper Bounds for Streett and Parity Complementation. Proceedings of the 20th Conference on Computer Science Logic (CSL 2011), Dagstuhl Publishing, 2011: 112-128.

\bibitem{CZ11(2)} Y. Cai, T. Zhang: A tight lower bound for Streett complementation. FSTTCS 2011: 339-350.








\bibitem{Piterman07} Piterman, N.: From nondeterministic B\"{u}chi and Streett automata
to deterministic parity automata. \emph{Journal of Logical Methods in
Computer Science} 3 (2007)




\bibitem{Sven09} Sven Schewe: Tighter Bounds for the Determinisation of B\"{u}chi Automata. FOSSACS 2009: 167-181

\bibitem{ScheweV14} S. Schewe, T. Varghese: Determinising Parity Automata. MFCS 2014: 486-498

\bibitem{Mic88} M. Michel. Complementation is more difficult with automata on infinite words. CNET, Paris, 1988.

\bibitem{Lod99} C. L\"{o}ding. Optimal bounds for the transformation of $\omega$-automata. In \emph{Proc. 19th Conf. on Foundations of Software Technology and Theoretical Computer Science}, volume 1738 of \emph{Lecture Notes in Computer Science}, pages 97-109, 1999.

\bibitem{CZ09} Thomas Colcombet, Konrad Zdanowski: A Tight Lower Bound for
Determinization of Transition Labeled B\"{u}chi Automata. ICALP 2009: 151-162

\bibitem{Yan06} Qiqi Yan. Lower bounds for complementation of omega-automata via the
full automata technique. ICALP 2006: 589-600.

\end{thebibliography}
\end{document}